\pdfoutput=1
\documentclass[letterpaper,10pt]{IEEEtran}
\usepackage{amssymb,amsmath,mathtools,graphicx,color}
\usepackage{enumerate}
\usepackage[symbol]{footmisc}

\usepackage{url}
\usepackage{cite}
\IEEEoverridecommandlockouts
\usepackage{amsthm}
\newtheorem{theorem}{Theorem}
\newtheorem{problem}{Problem}
\usepackage{algorithm}
\usepackage{cite}
\usepackage{amsmath}
\usepackage{multirow}
\usepackage{graphicx}
\usepackage{epstopdf}
\usepackage{epsfig}
\usepackage{bm}
\usepackage{pifont}
\usepackage{verbatim}
\usepackage{bigstrut,multirow,rotating}
\usepackage{booktabs}
\usepackage{multicol}
\usepackage{color}
\usepackage{mathrsfs}
\usepackage{textcomp}
\usepackage{mathrsfs,soul}
\newtheorem{assumption}{Assumption}

\graphicspath{{fig/}}
\usepackage{tikz}
\usepackage{textcomp}
\usepackage{hyperref}

\title{\huge Data-Driven Multi-Objective Controller Optimization for a Magnetically-Levitated Nanopositioning System}

\author{
 Xiaocong~Li,~\IEEEmembership{Member,~IEEE},
 Haiyue~Zhu,~\IEEEmembership{Member,~IEEE},
 Jun~Ma,~\IEEEmembership{Member,~IEEE},
 Tat~Joo~Teo,~\IEEEmembership{Member,~IEEE},
 Chek~Sing~Teo,~\IEEEmembership{Member,~IEEE},
	Masayoshi Tomizuka, \IEEEmembership{Life Fellow,~IEEE,}
 and Tong~Heng~Lee	
\thanks{This work is supported by Collaborative Research Project U18-R-030SU under SIMTech-NUS Joint Lab on Precision Motion Systems (U12-R-024JL).\textit{(Corresponding author: Haiyue Zhu.)}}
\thanks{X.~Li, H.~Zhu and C.~S. Teo are with the Mechatronics Group, Singapore Institute of Manufacturing Technology, Singapore 138634. (e-mail: li\_xiaocong@simtech.a-star.edu.sg, zhu\_haiyue@simtech.a-star.edu.sg, csteo@simtech.a-star.edu.sg).}
\thanks{J.~Ma and M.~Tomizuka
are with the Department of Mechanical Engineering, University of California, Berkeley, CA 94720, USA (e-mail: jun.ma@berkeley.edu, tomizuka@berkeley.edu).}
\thanks{T.~J. Teo and T.~H. Lee are with the Department of Electrical and Computer Engineering, National University of Singapore, Singapore 117583 (e-mail: eledttj@nus.edu.sg, eleleeth@nus.edu.sg).}
}

\newcommand\copyrighttext{%
	\scriptsize  \textbf{Accepted final version.} To appear in \emph{IEEE/ASME Transactions on Mechatronics}, DOI: 10.1109/TMECH.2020.2999401 
	\textcopyright 2020 IEEE.  Personal use of this material is permitted.  Permission from IEEE must be obtained for all other uses, in any current or future media, including reprinting/republishing this material for advertising or promotional purposes, creating new collective works, for resale or redistribution to servers or lists, or reuse of any copyrighted component of this work in other works.}
\newcommand\copyrightnotice{%
	\begin{tikzpicture}[remember picture,overlay]
	\node[anchor=south,yshift=10pt] at (current page.south) {\fbox{\parbox{\dimexpr\textwidth-\fboxsep-\fboxrule\relax}{\copyrighttext}}};
	\end{tikzpicture}%
}
\begin{document}	
\maketitle
\pagestyle{empty}
\thispagestyle{empty}
\copyrightnotice

\begin{abstract}
The performance achieved with traditional model-based control system design approaches
typically relies heavily upon accurate modeling of the motion dynamics.
However, modeling the true dynamics of present-day increasingly complex systems
can be an extremely challenging task;
and the usually necessary practical approximations
often renders the automation system to operate in a non-optimal condition.
This problem can be greatly aggravated in the case of
a multi-axis magnetically-levitated (maglev) nanopositioning system
where the fully floating behavior and multi-axis coupling
make extremely accurate identification
of the motion dynamics largely impossible.
On the other hand, in many related industrial automation applications, e.g.,
the scanning process with the maglev system,
repetitive motions are involved which could generate
a large amount of motion data under non-optimal conditions.
These motion data essentially contain rich information;
therefore, the possibility exists to develop an intelligent automation system
to learn from these motion data, and
to drive the system to operate towards optimality in a data-driven manner.
Along this line then, this paper proposes a data-driven model-free controller optimization approach
that learns from the past non-optimal motion data to iteratively improve the motion control performance.
Specifically, a novel data-driven multi-objective optimization approach
is proposed that is able to automatically estimate
the gradient and Hessian purely based on the measured motion data; the multi-objective cost function is suitably designed to take into account both smooth and accurate trajectory tracking.
In the work here, experiments are then conducted on the maglev nanopositioning system
to demonstrate the effectiveness of the proposed method,
and the results
show rather clearly
the practical appeal of our methodology for related complex robotic systems with no accurate model available.

\end{abstract}
\section{Introduction}

\IEEEPARstart{T}{he} magnetically-levitated nanopositioning technique~\cite{Zhu2019Tmech,Nguyen2017} is a promising solution for ultra-clean or vacuum precision motion applications
due to its excellent characteristics such as multi-axis mobility, ultra-precision,
large motion stroke, contact- and dust-free usage, etc.
However, due to its fully floating feature,
the maglev nanopositioning system requires sophisticated motion control
in all its six Degree-of-Freedom (DOF) to even simply stabilize at a constant position.
The advanced multi-axis positioning and trajectory tracking
further require high-performance precision motion control techniques \cite{kang2020six,lu2020adaptive,Chen2019precision,fan2019design,mishra2018precision}
to reject the internal/external disturbances and eliminate the coupling effects between axes.
Traditionally, such precision motion control systems
are designed and optimized based on the model (when available) obtained
from the first principle or system identification,
i.e., model-based approach \cite{wen2019an,xu2017continuous,lee2019harmonic,wang2019dynamical}.
However, obtaining an accurate model for the multi-axis maglev nanopositioning system
is challenging and time-consuming;
and the model obtained is typically often not adequately representative of the true dynamics,
e.g., the coupling between axes is often not taken into account.
To address this often-occurring general issue, it is a notable trend where
learning-based methods are increasingly being explored in literature, wherein model parameters are not precisely known,
and yet the
appropriately optimal control performance can be obtained \cite{hou2017data,yin2015data,ma2020advanced}.
This data-driven methodology enables learning from available signals in the past, and also prevailing,
non-optimal control settings
to achieve a significant performance improvement
for various cases of
real-world mechatronic systems \cite{Chen2020iterative,de2019finite,xie2019iterative,meng2017robust}.

Data-driven controls are essentially developed based on the concept
that machines can improve their performance by learning
from previous executions of the same or similar tasks,
in a way that closely resembles how humans learn.
A promising trend in data-driven controls is deep reinforcement learning \cite{hwangbo2019learning},
wherein a neural network policy is trained based on real-world motion data as well as appropriate simulations.
In addition to this end-to-end approach, deep neural networks
can also be used for trajectory tracking in many robotic applications \cite{li2017deep}.
However, the limitation of these neural network based approaches is the requirement for a massive amount of training data.
Also, the uncertainty in the important system stability issue due to the black-box nature of neural networks
often becomes a concern, especially for safety-critical applications.
Apart from the neural network based approaches,
\cite{yip2014model} proposed a novel model-less feedback control design for soft robotics,
and it was further extended to the hybrid position/force control problem in \cite{yip2016model}.
The proposed approaches in these works
allow the manipulators to interact with several constrained environments safely and stably,
and then generate a model-less feedback control policy from these interactions.
It is worthwhile to note that these works are mainly focused
on kinematic-model-free control
instead of dynamic-model-free control, i.e., the Jacobian is unknown and empirically estimated; therefore, challenges in dynamic control remain.

It is worthwhile to note that many industrial processes
such as scanning, pick-and-place, welding, and assembly,
involve repetitive motions; therefore, less computationally expensive learning approaches can be pursued.
For instance, the Iterative Learning Control (ILC) is a data-driven method
that is used widely in precision machines \cite{blanken2016batch, zhu2019internal}
and robotics \cite{kocc2019optimizing,angelini2018decentralized,tang2019model,wang2016robust}.
It makes use of the repetitive tracking error data gathered in previous cycles
to improve the performance of the system in subsequent cycles in a feedforward manner.
Thus, it is essentially a feedforward learning approach rather than feedback learning; nevertheless it can serve as a very useful complement to an existing feedback controller.
In \cite{berkenkamp2016safe},
the authors proposed a novel Gaussian process based feedback controller optimization algorithm
with applications to quadrotors.
This approach models the cost function as a Gaussian process
and explores the new controller parameters with a safe performance guarantee.
This enables automatic and safe optimization in repetitive robotic tasks without human intervention.
However, while greatly effective especially in guaranteeing safety, the convergence is relatively slow as it takes about 30 iterations to converge.

The Iterative Feedback Tuning (IFT) methodology is
one of the approaches in the class of
fast-converging data-driven controller optimization algorithms \cite{hjalmarsson2002iterative}.
Conceptually similar to the other approaches,
it makes use of the actual motion data to estimate the cost function gradient without relying on the system model.
In addition, the Hessian of the cost function can be estimated
to speed up the convergence.
The estimated gradient and Hessian are subsequently used in
the Gauss-Newton optimization procedure to iteratively obtain the optimal controller parameters.
This IFT approach has been widely used in many applications
such as path-tracking control of industrial robots \cite{xie2019iterative, xie2019robust},
ultra-precision wafer stage \cite{li2019convergence,heertjes2016constrained},
flow control over a circular cylinder \cite{son2018iterative}
and compliant rehabilitation robots \cite{meng2017robust}, etc.
Extensions of the IFT idea to other types of controller
includes iterative dynamic decoupling control \cite{li2019data}, disturbance observer sensitivity shaping \cite{li2019enhanced}, iterative feedforward tuning \cite{stearns2008iterative,van2008fixed}, and 3-DOF controller tuning \cite{li2017dataTIE,li2018data} etc. However, most of the existing work focused mainly on accurate tracking and did not take smooth tracking into account. In fact, in semiconductor manufacturing and many other robotic applications,
both accurate and smooth trajectory tracking are required \cite{ma2019parameter,ma2019parameterTIE}, and this challenge remains unsolved. Hence, the contribution of this paper is to propose a learning-based controller optimization algorithm to enable smooth and accurate tracking in repetitive tasks as illustrated schematically and conceptually in Fig.~\ref{fig:Overview}. To the best of our knowledge, this work is the first feedback controller optimization method to take into account both accurate and smooth tracking in a data-driven manner. Furthermore, it is worthwhile to note that the optimization process is both data-efficient and fast-converging.

\begin{figure}
	\centering
	\includegraphics[width=0.45\textwidth]{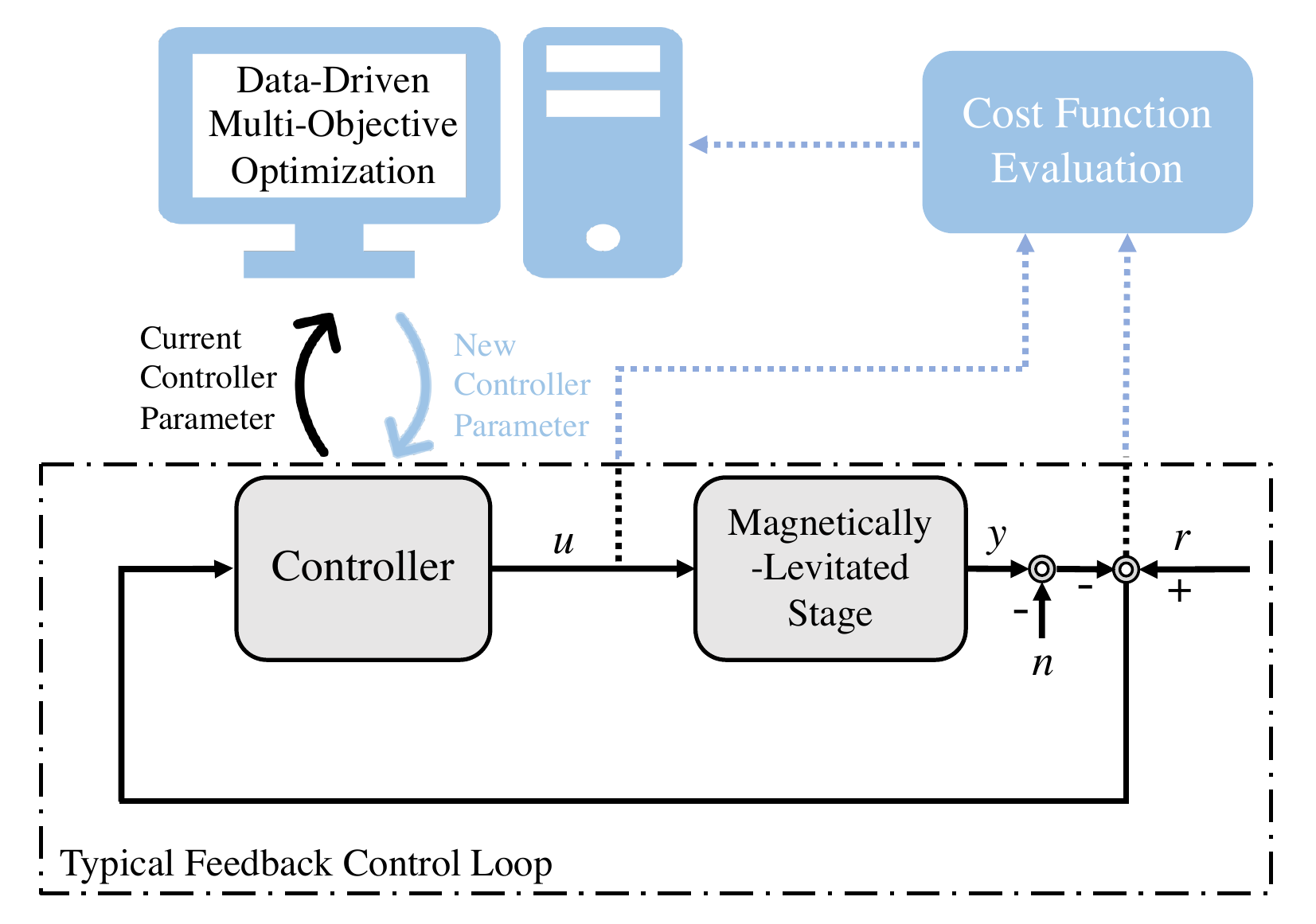}\\
	\caption{Overview of the data-driven multi-objective controller optimization algorithm. The algorithm iteratively updates the controller parameters based on the actual motion data from the previous iteration to minimize the cost function value in an iterative and model-free manner.}\label{fig:Overview}
\end{figure}

This paper is organized as follows. In Section II, a brief description of the magnetically-levitated nanopositioning system is provided. Then, in Section III, the proposed multi-objective controller optimization algorithm is described and analyzed in detail. In Section IV, experimental work is conducted based on the magnetically-levitated nanopositioning system to show the effectiveness of the proposed algorithm. Finally, conclusions are drawn in Section V.

\section{Magnetically-Levitated Nanopositioning System}
\begin{figure}
	\centering
	\includegraphics[width=0.45\textwidth]{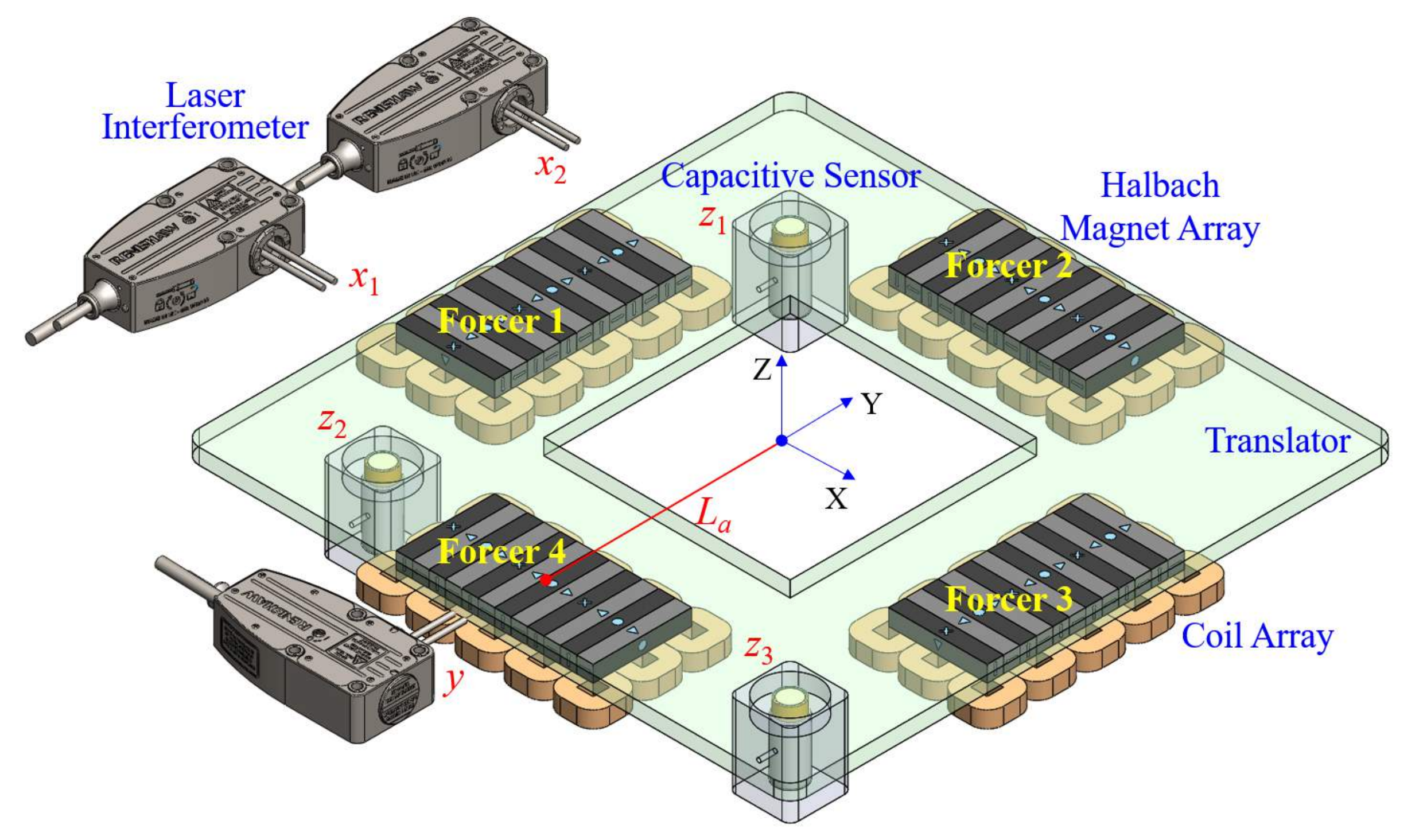}\\
	\caption{Schematics of the square-coil-based magnetically-levitated planar nanopositioning system.} \label{fig:MaglevSchematics}
\end{figure}
In this section, the magnetically-levitated planar nanopositioning system
(which is the typical prototype application
of our data-driven controller optimization approach)
is first illustrated, including its working principles and associated overall control scheme.
The design objective of our magnetically-levitated planar nanopositioning system
is to enable 6-DOF motion with low system complexity and high energy efficiency.
For large-stroke applications, the stroke expandability is also important
as well as the affordability to simultaneously operate multiple motion translators.
The schematic design of the implemented magnetically-levitated planar nanopositioning system for this work
is illustrated in Fig.~\ref{fig:MaglevSchematics}.
Although the square coil array in Fig.~\ref{fig:MaglevSchematics} is covered here in a small area for evaluation,
such a square coil based design allows suitably unlimited planar motion stroke as long as the coils spread over.
Notably, this system adopts the tiled square coil array for actuation,
which shows the comparative advantages in control complexity and energy efficiency as it only requires 8-phase for 6-DOF motion control and coils far away can be actively switched off to save energy \cite{zhu2016design}.
Furthermore, the interference between coils is minimized at the maximum extent
by using the square coil arrangement,
so that multiple translators are feasible by individually controlling each or set of coils.

\begin{figure}
	\centering
	\includegraphics[width=0.45\textwidth]{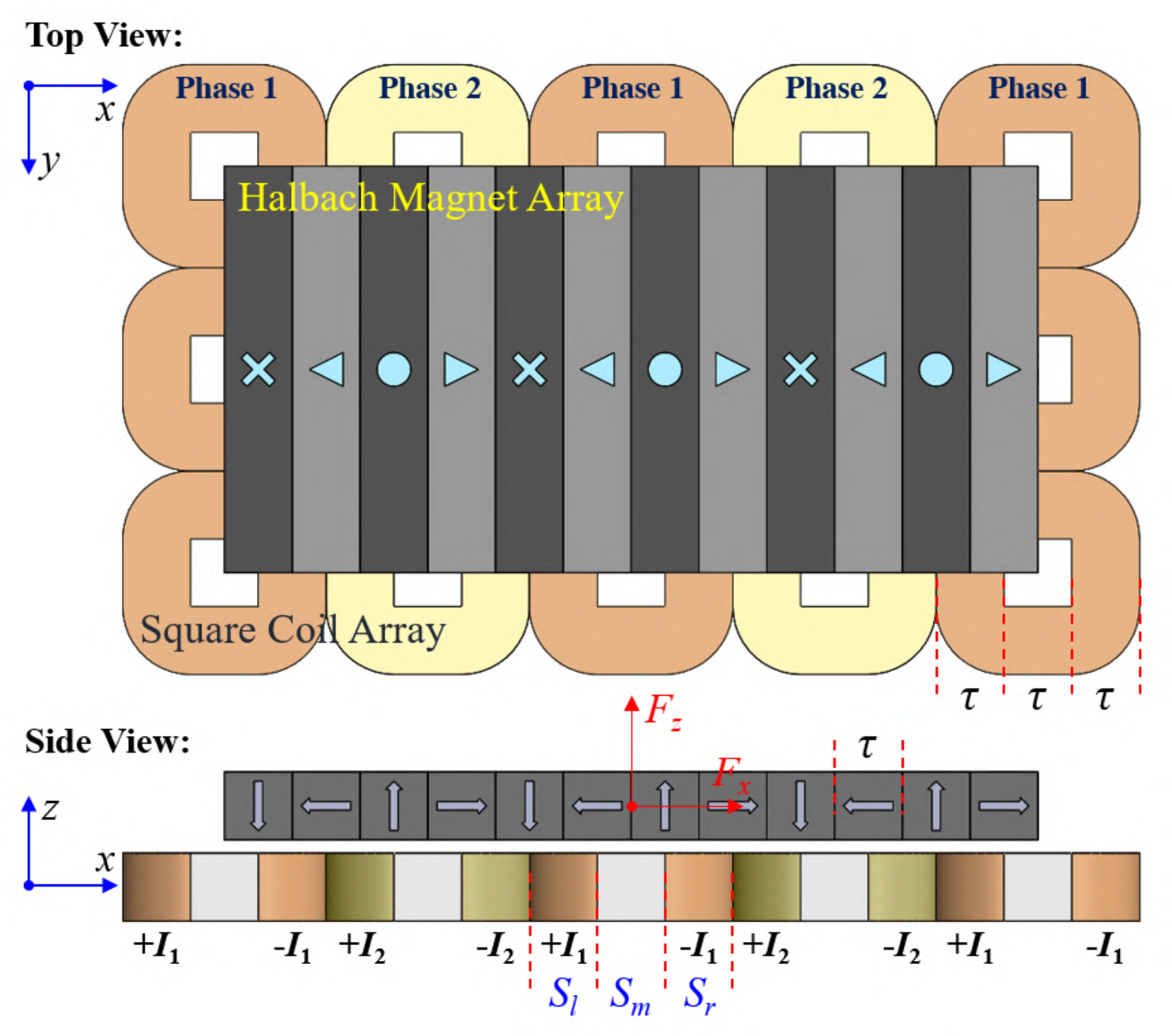}\\
	\caption{Top view and side view of the forcer in the square-coil-based magnetically-levitated stage design.} \ \label{fig:SchematicsForcer}
\end{figure}
From Fig.~\ref{fig:MaglevSchematics}, it can be seen that
the 6-DOF motion is achieved by the combined force from Forcer 1 to 4,
where each forcer can provide a vertical levitation force and horizontal thrust force.
As illustrated in Fig.~\ref{fig:SchematicsForcer},
the moving part of one forcer is a Halbach permanent magnet array
and the stationary part is a square coil array grouped into two phases.
Due to the periodic arrangement of magnetization directions indicated in Fig.~\ref{fig:SchematicsForcer}, the Halbach array generates an almost ideal sinusoidal magnetic field
in both X and Z axes except the magnet end effects. This is not achievable with the normal magnet array widely adopted in 1-DOF linear motors \cite{zhu2016design}. From the side view of Fig.~\ref{fig:SchematicsForcer},
each square coil is divided into three segments, i.e., $S_l$, $S_m$, and $S_r$.
The current directions in $S_l$ and $S_r$ are opposite
and the magnetic field directions for $S_l$ and $S_r$ are also opposite,
so that $S_l$ and $S_r$ generate identical force
in both X and Z axes. $S_m$ contributes zero force in two axes due to its current direction.

The force generation on a single square coil can be expressed via the relative location between coil and magnet arrays $(x,z)$ as $F^{c}_{x}(x,z)=K_x(x,z)I$ and $F^{c}_{z}(x,z)=K_z(x,z)I$, where $I$ is the current magnitude, and $K_x(x,z)$ and $K_z(x,z)$ are defined as,
\begin{equation}\label{KxKz}
\begin{split}
K_x(x,z)\!=&C_{f}e^{-\gamma z}\!\Big(\!-\!\sqrt{2}\tau\sin(\gamma x)\!+\!\frac{2\alpha}{\gamma}\sin(\gamma x)\cos( \frac{\gamma\tau}{2}\!+\!\beta)\!\Big), \\
K_z(x,z)\!=&C_{f}e^{-\gamma z}\!\Big(\!-\!\sqrt{2}\tau\cos(\gamma x)\!+\!\frac{2\alpha}{\gamma}\cos(\gamma x)\cos( \frac{\gamma\tau}{2}\!+\!\beta)\!\Big), \\
\end{split}
\end{equation}
where $C_{f}>0$ is a force constant, $\tau$ is the geometrical dimension as indicated in Fig.~\ref{fig:SchematicsForcer},
$\gamma$ is the spatial wave number with $\gamma=\pi/2\tau$,
and
$\alpha$ and $\beta$ are two constants numerically identified, $\alpha=1.144$ and $\beta=2.3924$.
Therefore, the total force generated by the whole forcer in Fig.~\ref{fig:SchematicsForcer} is expressed via two phases of current as,
\begin{equation}\label{FxFz}
\begin{split}
\bm{F^{f}}(x,z)
& = N
\left[
\begin{array}{cc}
K_{x}(x,z) & K_{x}(x+3\tau,z) \\
K_{z}(x,z) & K_{z}(x+3\tau,z)
\end{array}
\right]
\bm{I}
\end{split}
\end{equation}
where $\bm{F^{f}}(x,z)=[F^{f}_{x}(x,z)\ F^{f}_{x}(x,z)]^{T}$, $\bm{I}=[I_1\ I_2]^{T}$, $I_1$ and $I_2$ denote the current magnitudes in Phase 1 and Phase 2, respectively. $N$ is the number of effective coils in each phase, where $N=4$ for the case in Fig.~\ref{fig:SchematicsForcer}, and denotes
\begin{equation}\label{Phi}
\begin{split}
\bm{\Phi_{K}}(x,\,z)=
\left[
\begin{array}{cc}
K_{x}(x,z) & K_{x}(x+3\tau,z) \\
K_{z}(x,z) & K_{z}(x+3\tau,z)
\end{array}
\right]
\end{split}
.
\end{equation}

In order to control the 6-DOF motion, the global force/torque given by the controller needs to be allocated to four forcers. For each forcer, such local force is generated through energizing the two-phase current $\bm{I_{i}}=[I_{i1}\ I_{i2}]^{T}$ on each forcer. According to~\eqref{FxFz},
\begin{equation}\label{FxFzInverse}
\begin{aligned}
\bm{I_{i}}=\bm{\Phi_{K}}(x,\,z)^{-1}\bm{F^{f_{i}}}/N.
\end{aligned}
\end{equation}
Therefore, the controllability of the square coil magnetically-levitated system design is based on the invertibility of $\bm{\Phi_{K}}(x,\,z)$. It is noted that
\begin{equation}\label{Invertibility}
\begin{aligned}
&\det\Big(\bm{\Phi_{K}}(x,\,z)\Big) \\
=&\,K_{x}(x,z) K_{z}(x+3\tau,z) - K_{z}(x,z)K_{x}(x+3\tau,z)\\
=&\,\Big(C_{f}e^{-\gamma z}\Big)^2\Bigg(\Big(-\!\sqrt{2}\tau\sin(\gamma x)+\frac{2\alpha}{\gamma}\sin(\gamma x)\cos( \frac{\gamma\tau}{2}+\beta)\Big)^2 \\
& + \Big(-\!\sqrt{2}\tau\cos(\gamma x)+\frac{2\alpha}{\gamma}\cos(\gamma x)\cos( \frac{\gamma\tau}{2}+\beta)\Big)^2\Bigg)\\
=&\,\Big(C_{f}e^{-\gamma z}\Big)^2\Big(\sqrt{2}\tau-\frac{2\alpha}{\gamma}\cos( \frac{\gamma\tau}{2}+\beta)\Big)^2.
\end{aligned}
\end{equation}
Since $\gamma=\pi/2\tau$, and thus $\gamma\tau=\pi/2$, it can be seen that as long as $\cos( \beta+\pi/4)\neq{\sqrt{2}\pi}/{4\alpha}$,
$\det\Big(\bm{\Phi_{K}}(x,\,z)\Big)>0$
indicates that $\bm{\Phi_{K}}(x,\,z)$ is full rank and invertible,
with the values of the position $(x,\,z)$ not affecting this property.
Numerically, as the values of $\alpha$ and $\beta$ are known,
it is thus direct to verify that the condition $\cos( \beta+\pi/4)\neq{\sqrt{2}\pi}/{4\alpha}$ is met,
which shows that~\eqref{FxFzInverse}
has no singularity and the 6-DOF motion is fully controllable.
The 6-DOF sensing is achieved via three channels of laser interferometers $(x_1,\,x_2,\,y)$
and three channels of capacitive sensors $(z_1,\,z_2,\,z_3)$ as indicated in Fig.~\ref{fig:MaglevSchematics}.
With the measured 6-axis state-variables, each DOF can thus be closed-loop controlled as Single-Input Single-Output (SISO) systems and ready for the deployment of the algorithm in Section \ref{algorithm}.

\section{Data-Driven Multi-Objective Optimization}\label{algorithm}

As noted earlier, certain important precision motion systems such as the maglev nanopositioning system
emphasize the requirement for smooth and accurate tracking in terms of control performance.
To achieve these objectives, both the tracking accuracy and control signal variation
needs to be taken into account concurrently in the optimization.
Hence, the overall cost function in this paper is defined as
\begin{equation}
J(^\mathbf{i}\rho)=\underbrace{w_1e(^\mathbf{i}\rho)^T \cdot e(^\mathbf{i}\rho)}_{J_e}+\underbrace{w_2\dot u(^\mathbf{i}\rho)^T \cdot \dot u(^\mathbf{i}\rho)}_{J_{\dot u}},\label{eq:J}
\end{equation}
where $^\mathbf{i}\rho$ is the controller parameter vector in the $\mathbf{i}^{\textrm{th}}$ iteration,
and $J(^\mathbf{i}\rho)$ is the total cost function consisting
of the tracking related cost function $J_e$ and control variation related cost function $J_{\dot u}$.
Here, $w_1$ is the weighting for the tracking performance
wherein $e(^\mathbf{i}\rho)$ is the tracking error measured in the $\mathbf{i}^{\textrm{th}}$ iteration; $w_2$ is the weighting for the control variation
wherein $u(^\mathbf{i}\rho)$ is the control input
and $\dot u(^\mathbf{i}\rho)$ is the variation of control input.
Thus consider the typical feedback control system for the magnetically-levitated system
as in Fig. \ref{fig:Overview},
where a fixed structure controller $C(s)$ is used for motion control
and can be expressed as
\begin{equation}
C(s,\rho)=\rho^T \bar C(s).\label{eq:Fixed_structure_C}
\end{equation}
Here, $\rho$ is a vector of the controller parameters to be optimized
and $\bar C(s)$ is a vector of parameter independent transfer functions.
We can now formulate the data-driven multi-objective optimization problem as:
\begin{problem}
	Assume the motion system is unknown and controlled by a fixed structure controller $C(s,\rho)$ in \eqref{eq:Fixed_structure_C};
use only the closed-loop experimental data
to determine the parameter vector $\rho$ that minimizes
the multi-objective cost function $J(\rho)$ \eqref{eq:J}, i.e., to find
	\begin{equation}
	\rho^\star= \textrm{arg}\: \underset{\rho}{\textrm{min}} \:J(\rho).\label{eq:probformulation}
	\end{equation}
\end{problem}

\subsection{Gradient Calculation and Estimation}
With equation \eqref{eq:J}, the gradient of the cost function $J(^\mathbf{i}\rho)$ with respect to the parameter in the $\textrm{i}^\textrm{th}$ iteration $^\mathbf{i}\rho$ can be derived as
\begin{align}\label{eq:Gradient}
\nabla J(^\mathbf{i}\rho)&=2w_1[\nabla \: {^\mathbf{i}e}(^\mathbf{i}\rho)]^T \cdot {^\mathbf{i}e}(^\mathbf{i}\rho) \nonumber \\
&\quad +2w_2[\nabla \: {^\mathbf{i}\dot u}(^\mathbf{i}\rho)]^T \cdot {^\mathbf{i}\dot u}(^\mathbf{i}\rho),
\end{align}
and the Hessian of the cost function can be approximated as
\begin{align}\label{eq:Hessian}
\nabla^2 J(^\mathbf{i}\rho)&=2w_1[\nabla \: {^\mathbf{i}e}(^\mathbf{i}\rho)]^{T}\cdot \nabla \: {^\mathbf{i}e}(^\mathbf{i}\rho) \nonumber \\
&\quad +2w_2[\nabla \: {^\mathbf{i}\dot u}(^\mathbf{i}\rho)]^{T}\cdot \nabla \: {^\mathbf{i}\dot u}(^\mathbf{i}\rho).
\end{align}
The purpose of obtaining the gradient and the Hessian of the cost function is to apply the Newton's optimization algorithm \cite{boyd2004convex}:
\begin{equation}
{{^\mathbf{i+1}\rho}}={^\mathbf{i}\rho} - {^\mathbf{i}\gamma}(\nabla^2 J(^\mathbf{i}\rho))^{-1} \: \nabla J(^\mathbf{i}\rho).\label{eq:Newton}
\end{equation}
where ${{^\mathbf{i+1}\rho}}$ is the updated parameter value for iteration $\mathbf{i+1}$ and $^\mathbf{i}\gamma$ is the step size at iteration $\mathbf{i}$. From \eqref{eq:Gradient} and \eqref{eq:Hessian}, the Newton's optimization algorithm requires $\nabla \: {^\mathbf{i}e}(^\mathbf{i}\rho)$, $\nabla \: {^\mathbf{i}\dot u}(^\mathbf{i}\rho)$, ${^\mathbf{i}e}(^\mathbf{i}\rho)$ and ${^\mathbf{i}\dot u}(^\mathbf{i}\rho)$. ${^\mathbf{i}e}(^\mathbf{i}\rho)$ and ${^\mathbf{i}\dot u}(^\mathbf{i}\rho)$ can be obtained directly from the sensor measurement and the control software. However, $\nabla \: {^\mathbf{i}e}(^\mathbf{i}\rho)$ and $\nabla \: {^\mathbf{i}\dot u}(^\mathbf{i}\rho)$ cannot be obtained directly and have to be estimated with the input-output data collected from the closed-loop experiments.
The gradient of the tracking error can be derived as:
\begin{align}\label{eq:gradient_of_e}
\nabla {^\mathbf{i}e}(^\mathbf{i}\rho)=\frac{-P\frac{\partial C(^\mathbf{i}\rho)}{\partial ^\mathbf{i}\rho}} {[1+PC(^\mathbf{i}\rho)]^2}\cdot r\nonumber
= -\frac{P\frac{\partial C(^\mathbf{i}\rho)}{\partial ^\mathbf{i}\rho}}{1+PC(^\mathbf{i}\rho)} \cdot {^\mathbf{i}e}(^\mathbf{i}\rho).
\end{align}
Inspired by the IFT approach \cite{hjalmarsson2002iterative}, $\nabla {^\mathbf{i}e}(^\mathbf{i}\rho)$ can then be obtained by setting ${^\mathbf{i}e}(^\mathbf{i}\rho)$ as the new reference $r$ in the ``special" experiment, and we have
\begin{equation}
\nabla {^\mathbf{i}e}(^\mathbf{i}\rho)=-\frac{\partial C(^\mathbf{i}\rho)}{\partial ^\mathbf{i}\rho}\cdot \frac{1}{C(^\mathbf{i}\rho)} \cdot y_s,\label{eq:gradient_of_e_final}
\end{equation}
where $y_s$ denotes the position measurement for this experiment. Apart from $\nabla {^\mathbf{i}e}(^\mathbf{i}\rho)$, the gradient of ${^\mathbf{i}\dot u}(^\mathbf{i}\rho)$ can also be derived as
\begin{align}\label{eq:gradient_of_udot}
\nabla {^\mathbf{i}\dot u}(^\mathbf{i}\rho)&=\frac{\frac{\partial C(^\mathbf{i}\rho)}{\partial ^\mathbf{i}\rho}[1+PC(^\mathbf{i}\rho)]}{[1+PC(^\mathbf{i}\rho)]^2}\cdot \dot r -\frac{P\frac{\partial C(^\mathbf{i}\rho)}{\partial ^\mathbf{i}\rho}C(^\mathbf{i}\rho)}{[1+PC(^\mathbf{i}\rho)]^2} \cdot \dot r \nonumber \\
&= \frac{\partial C(^\mathbf{i}\rho)}{\partial ^\mathbf{i}\rho}\frac{1}{1+PC(^\mathbf{i}\rho)}\cdot \dot e
\end{align}
$\nabla {^\mathbf{i}\dot u}(^\mathbf{i}\rho)$ can be estimated with the same special experiment by feeding in ${^\mathbf{i}e}(^\mathbf{i}\rho)$ as the reference $r$
\begin{equation}
\nabla {^\mathbf{i}\dot u}(^\mathbf{i}\rho)=\frac{\partial C(^\mathbf{i}\rho)}{\partial ^\mathbf{i}\rho}\cdot \frac{1}{C(^\mathbf{i}\rho)} \cdot \dot u_s,\label{eq:gradient_of_udot_final}
\end{equation}
where $u_s$ denotes the control input of this special experiment. Notice that $\nabla {^\mathbf{i}e}(^\mathbf{i}\rho)$ and $\nabla {^\mathbf{i}\dot u}(^\mathbf{i}\rho)$ can be estimated solely based on the experimental data. In addition, ${^\mathbf{i}e}(^\mathbf{i}\rho)$ and ${^\mathbf{i}\dot u}(^\mathbf{i}\rho)$ can be directly obtained or calculated based on the sensor measurement and control software. Hence, the gradient $\nabla J(^\mathbf{i}\rho)$ and Hessian $\nabla^2 J(^\mathbf{i}\rho)$ of the cost function can also be estimated according to \eqref{eq:Gradient} and \eqref{eq:Hessian}. It should be noted, as will be discussed in Section \ref{exp} and \ref{unbiase}, that an additional normal experiment needs to be conducted in order to obtain an unbiased estimate of the gradient when the measurement noise is taken into consideration.
\subsection{Data Collection}\label{exp}
To make the data-driven optimization procedure clearer, all the experiments needed and data to be collected within a single iteration are listed below.
\begin{itemize}
	\item Experiment I: Normal experiment.
	\begin{eqnarray}
	&r^\mathbf{1}=r, \label{eq:E1_1}\\
	&y^\mathbf{1}=\frac{PC(^\mathbf{i}\rho)}{1+PC(^\mathbf{i}\rho)}\cdot r^\mathbf{1}- \frac{1}{1+PC(^\mathbf{i}\rho)} \cdot n^\mathbf{1},\\
	&e^\mathbf{1}=\frac{1}{1+PC(^\mathbf{i}\rho)}\cdot r^\mathbf{1} + \frac{1}{1+PC(^\mathbf{i}\rho)} \cdot n^\mathbf{1}.
	\end{eqnarray}
	\item Experiment II: Special experiment.
	\begin{eqnarray}
	&r^\mathbf{2}= e^\mathbf{1},\label{eq:E2_1}\\
	&y_s=y^\mathbf{2}= \frac{PC(^\mathbf{i}\rho)}{1+PC(^\mathbf{i}\rho)}\cdot e^\mathbf{1} - \frac{1}{1+PC(^\mathbf{i}\rho)} \cdot n^\mathbf{2},\\
	&u_s=u^\mathbf{2}= \frac{C(^\mathbf{i}\rho)}{1+PC(^\mathbf{i}\rho)}\cdot e^\mathbf{1}+\frac{C(^\mathbf{i}\rho)}{1+PC(^\mathbf{i}\rho)}\cdot n^\mathbf{2} .
	\end{eqnarray}
	\item Experiment III: Normal experiment.
	\begin{eqnarray}
	&r^\mathbf{3}=r,\label{eq:E3_1}\\
	&e^\mathbf{3}=\frac{1}{1+PC(^\mathbf{i}\rho)}\cdot r^\mathbf{3} + \frac{1}{1+PC(^\mathbf{i}\rho)} \cdot n^\mathbf{3}.\\
	&u^\mathbf{3}= \frac{C(^\mathbf{i}\rho)}{1+PC(^\mathbf{i}\rho)}\cdot r^\mathbf{3}+\frac{C(^\mathbf{i}\rho)}{1+PC(^\mathbf{i}\rho)}\cdot n^\mathbf{3}.
	\end{eqnarray}
\end{itemize}

The bold right superscript refers to the experiment index within a single iteration. In Experiment I, the normal operation with, e.g., a S-curve trajectory, is conducted while $y^\mathbf{1}$ is measured and used to generate $e^\mathbf{1}$ as the reference of Experiment II. In Experiment II, measurement of $y_s$ and $u_s$ is taken and it is then used to obtain $\nabla {^\mathbf{i}e}(^\mathbf{i}\rho)$ and $\nabla {^\mathbf{i}\dot u}(^\mathbf{i}\rho)$ according to \eqref{eq:gradient_of_e_final} and \eqref{eq:gradient_of_udot_final}. In Experiment III, measurement of $e^\mathbf{3}$ and $u^\mathbf{3}$ is taken and used to calculate the cost function gradient $\nabla J(^\mathbf{i}\rho)$. The complete data-driven multi-objective optimization algorithm can be summarized in \textbf{Algorithm 1}. It is worth noting that, similar to the IFT and many other algorithms inspired by the IFT, there is no strong guarantee (proofs) for robust stability throughout the iterations, due to the lack of the system model. Hence, as also suggested in\mbox{ \cite{hjalmarsson2001robust}}, we shall use cautious updates, i.e., use small step-sizes, especially during the first iterations.

\begin{algorithm}[H]
	\caption{Data-Driven Multi-Objective Controller Optimization Algorithm}
	\begin{enumerate}
		\item Set the iteration number $\mathbf{i}=0$ and select the initial controller parameter ${^\mathbf{0}\rho}$.
		\item Conduct Experiment I and measure the output $y^\mathbf{1}$ and tracking error $e^\mathbf{1}$.
		\item Evaluate the cost function $J({^\mathbf{i}\rho}).$
		Stop if the cost function value is satisfactory. Otherwise, proceed to Step 4.
		\item Conduct Experiment II and measure the output $y^\mathbf{2}$ from this special experiment.
		\item Obtain $\nabla \: {^\mathbf{i}e}({^\mathbf{i}\rho})$ and $\nabla {^\mathbf{i}\dot u}(^\mathbf{i}\rho)$ according to \eqref{eq:gradient_of_e_final} and \eqref{eq:gradient_of_udot_final} respectively.
		\item Conduct Experiment III and measure $e^\mathbf{3}$ and $u^\mathbf{3}$.
		\item Compute $\nabla J({^\mathbf{i}\rho})$ as well as $\nabla^2 J({^\mathbf{i}\rho})$ according to \eqref{eq:Gradient} and \eqref{eq:Hessian}, where ${^\mathbf{i}e}({^\mathbf{i}\rho})$, ${^\mathbf{i}\dot{u}}({^\mathbf{i}\rho})$ are obtained from Experiment III and $\nabla \: {^\mathbf{i}e}({^\mathbf{i}\rho})$,$\nabla \: {^\mathbf{i}\dot{u}}({^\mathbf{i}\rho})$ are obtained from Step 5.
		\item Execute the Gauss-Newton algorithm \eqref{eq:Newton}, and update the controller parameters.
		\item Set the iteration number $\mathbf{i}\gets\mathbf{i}+1$ and proceed to Step 2.
	\end{enumerate}
	\label{algo1}
\end{algorithm}
\subsection{Unbiasedness of the Gradient Estimation}\label{unbiase}
The cost function gradient is estimated using the closed-loop experiment data, so the measurement noises can potentially lead to errors during this estimation. For this stochastic approximation method to work, the gradient estimation has to be unbiased, mathematically
\begin{eqnarray}
E\{\textrm{est}[\nabla J(^\mathbf{i}\rho)]\} = \nabla J(^\mathbf{i}\rho).
\label{eq:Unbiase_def}
\end{eqnarray}
To prove the unbiasedness, we have the following assumptions:
\begin{assumption} \label{ass:1}
	Noises $n$ in different experiments are independent from each other.
\end{assumption}
\begin{assumption} \label{ass:2}
	Noises $n$ are zero mean, weakly stationary random variables.
\end{assumption}

\begin{theorem} \label{th:1}
	For the motion system under the feedback control configuration as shown in Fig. \ref{fig:Overview}, with Assumption \ref{ass:1} and Assumption \ref{ass:2}, the estimation of the gradient of the cost function $J$ in \eqref{eq:J} is unbiased.
\end{theorem}
\begin{proof}
From \eqref{eq:gradient_of_e_final}, the estimated gradient of $e$ is given by

\begin{eqnarray}
\textrm{est}[\nabla {^\mathbf{i}e}(^\mathbf{i}\rho)]&=&-\frac{\partial C(^\mathbf{i}\rho)}{\partial\, ^\mathbf{i}\rho}\cdot \frac{P}{1+PC(^\mathbf{i}\rho)} \cdot e^\mathbf{1}\nonumber\\
&&-\frac{\partial C(^\mathbf{i}\rho)}{\partial\, ^\mathbf{i}\rho}\cdot \frac{P}{[1+PC(^\mathbf{i}\rho)]^2}\cdot n^\mathbf{1} \nonumber\\
&&+\frac{\partial C(^\mathbf{i}\rho)}{\partial\, ^\mathbf{i}\rho}\cdot \frac{1}{C(^\mathbf{i}\rho)} \cdot \frac{1}{1+PC(^\mathbf{i}\rho)} \cdot n^\mathbf{2}  \nonumber\\
&=& \nabla {^\mathbf{i}e}(^\mathbf{i}\rho) + w_e,
\label{eq:unbiase}
\end{eqnarray}
where
\begin{eqnarray}
w_e&\equiv& -\frac{\partial C(^\mathbf{i}\rho)}{\partial\, ^\mathbf{i}\rho}\cdot \frac{P}{[1+PC(^\mathbf{i}\rho)]^2}\cdot n^\mathbf{1} \nonumber\\
&&+\frac{\partial C(^\mathbf{i}\rho)}{\partial\, ^\mathbf{i}\rho}\cdot \frac{1}{C(^\mathbf{i}\rho)} \cdot \frac{1}{1+PC(^\mathbf{i}\rho)} \cdot n^\mathbf{2}.
\end{eqnarray}
Notice that $w_e$ contains noises from Experiment I and Experiment II and $e^\mathbf{3}$ contains only the noises from Experiment III. With Assumption \ref{ass:1} and Assumption \ref{ass:2}, we have
\begin{eqnarray}
E[{w_e}^T \cdot e^\mathbf{3}(^\mathbf{i}\rho)]=E[{w_e}^T] \cdot E[e^\mathbf{3}(^\mathbf{i}\rho)],\label{eq:independent}\end{eqnarray}
and
\begin{eqnarray}
E[{w_e}^T]=0.\label{eq:zeroexp}
\end{eqnarray}
Similar results can be obtained for $\dot u$ from \eqref{eq:gradient_of_udot_final}. The expectation of the estimation of the cost function gradient can be derived as follows
\begin{align}
E\{\textrm{est}[\nabla &J(^\mathbf{i}\rho)]\}\nonumber\\
=&2w_1E\{\textrm{est}[\nabla e^T(^\mathbf{i}\rho)]e^\mathbf{3}(^\mathbf{i}\rho)\}\nonumber\\
&+2w_2E\{\textrm{est}[\nabla {\dot u}^T(^\mathbf{i}\rho)]{\dot u}^\mathbf{3}(^\mathbf{i}\rho)\}\nonumber\\
=& 2w_1E[\nabla e^T(^\mathbf{i}\rho)e^\mathbf{3}(^\mathbf{i}\rho)]+2w_1E[{w_e}^T \cdot e^\mathbf{3}(^\mathbf{i}\rho)]\nonumber\\
&+2w_2E[\nabla {\dot u}^T(^\mathbf{i}\rho){\dot u}^\mathbf{3}(^\mathbf{i}\rho)]+2w_2E[{w_{\dot u}}^T \cdot {\dot u}^\mathbf{3}(^\mathbf{i}\rho)]\nonumber\\
=& \nabla J(^\mathbf{i}\rho)+0 \cdot E[e^\mathbf{3}(^\mathbf{i}\rho)]+0 \cdot E[{\dot u}^\mathbf{3}(^\mathbf{i}\rho)]\nonumber\\
=& \nabla J(^\mathbf{i}\rho).
\label{eq:unbiase_J}
\end{align}

This completes the proof of the Theorem.
\end{proof}

From the proof, it can be noticed that Experiment III is indeed necessary in order to guarantee the unbiasedness of cost function gradient estimation. If the data from Experiment I were used, i.e., $e^\mathbf{1}(^\mathbf{i}\rho)$ and ${\dot u}^\mathbf{1}(^\mathbf{i}\rho)$ instead of $e^\mathbf{3}(^\mathbf{i}\rho)$ and ${\dot u}^\mathbf{3}(^\mathbf{i}\rho)$, the same noise would exist in both $\textrm{est}[\nabla e^T(^\mathbf{i}\rho)]$ and $e(^\mathbf{i}\rho)$ (as well as in $\textrm{est}[\nabla {\dot u}^T(^\mathbf{i}\rho)]$ and ${\dot u}(^\mathbf{i}\rho)$). This would lead to a biased estimation for the cost function gradient and it is exactly the reason why Experiment III is needed.

\section{Experimental Validation}
\begin{figure}
	\centering
	\includegraphics[width=0.45\textwidth]{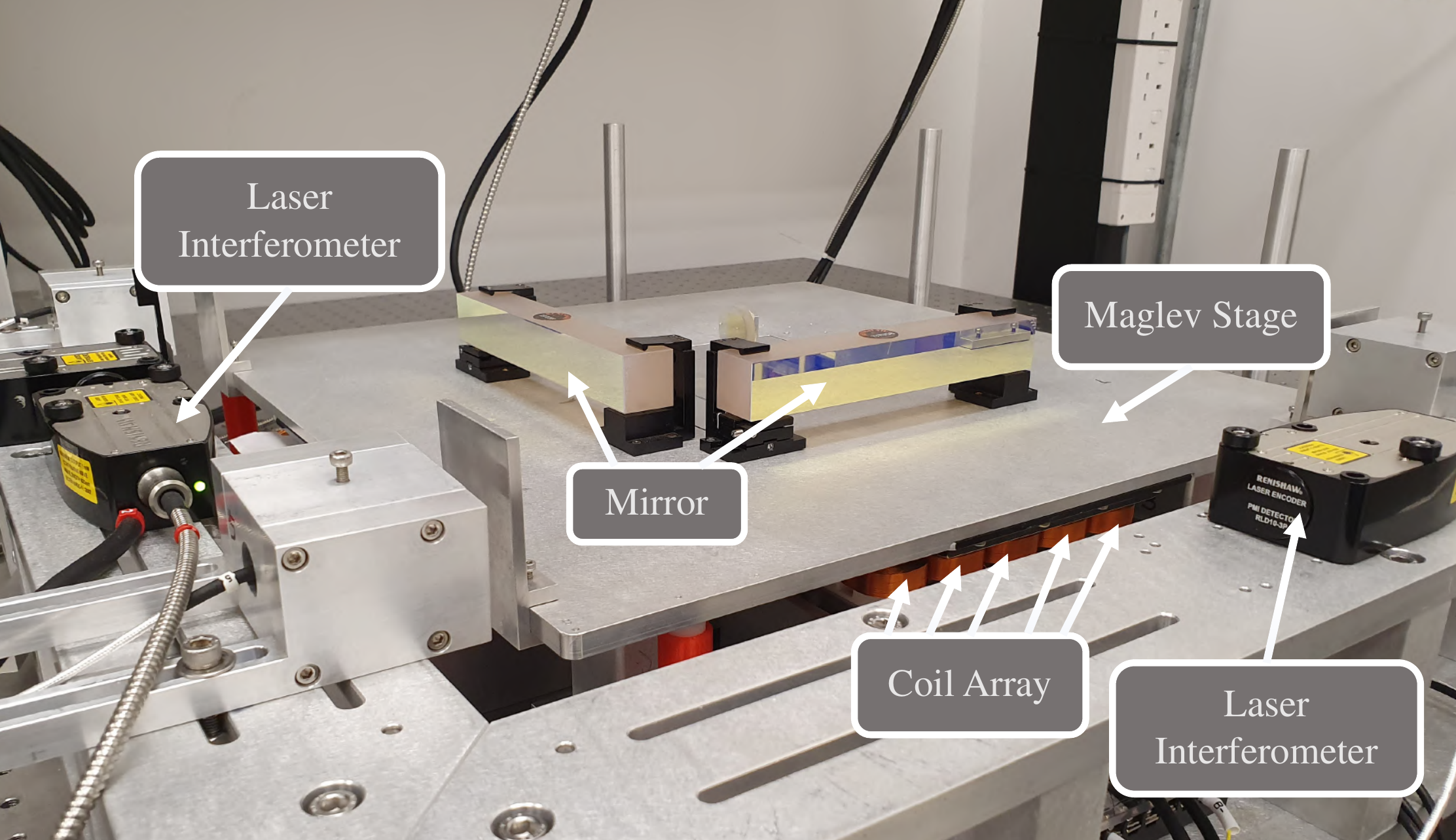}\\
	\caption{Magnetically-levitated nanopositioning system used in the experimental validation.}\label{fig:Maglevpic}
\end{figure}
\begin{figure}
	\centering
	\includegraphics[width=0.45\textwidth]{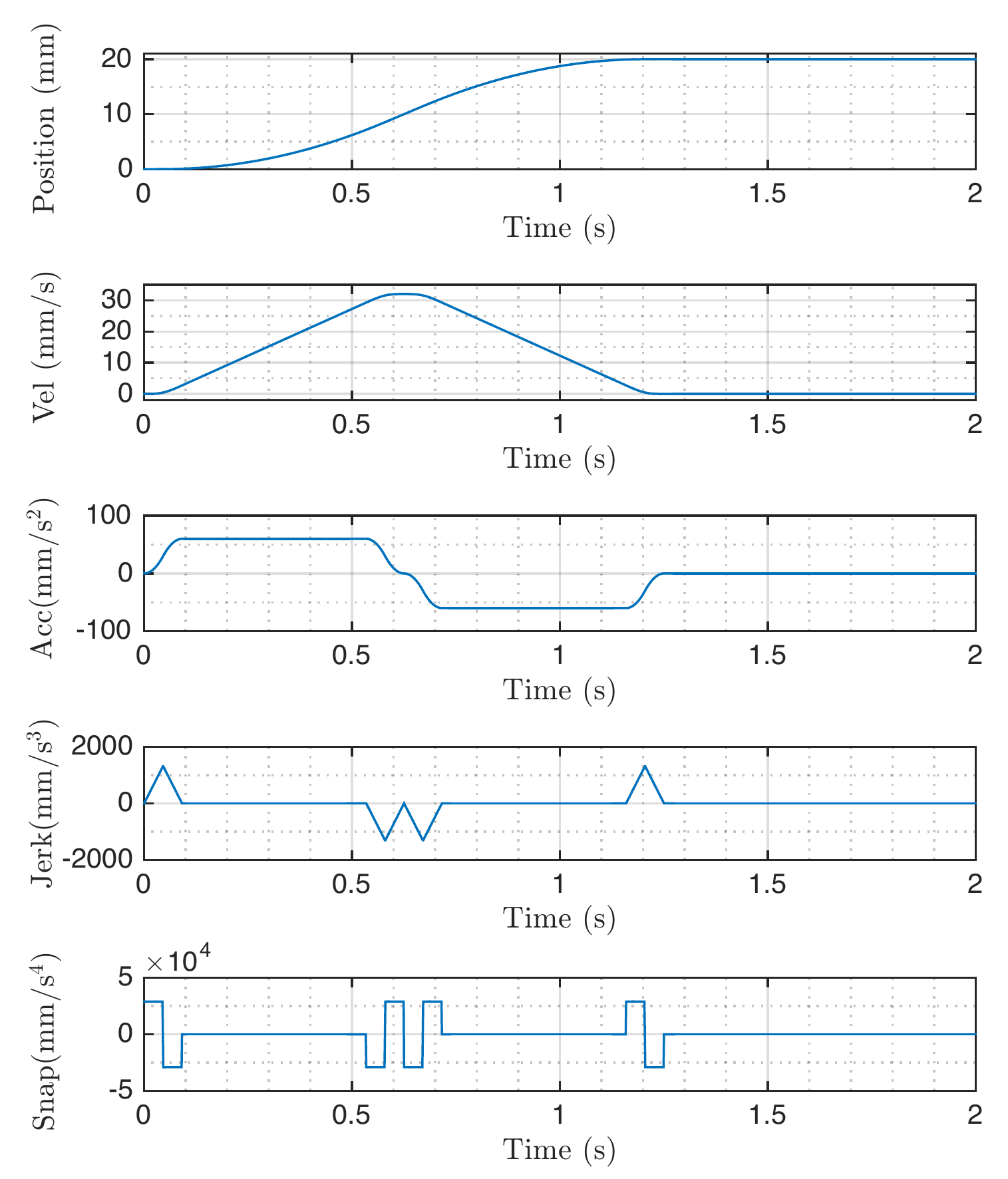}\\
	\caption{Fourth-order S-curve motion profile used in the real-time experiment.}\label{fig:Scurve}
\end{figure}

\begin{table*}[t]
	\caption{Overview of the controller parameters for Y and X axis}
	\begin{center}
		\begin{tabular}{ c|c|c|c|c }
			\multirow{2}{*}{Parameters} &  \multicolumn{2}{c|}{Y Axis} & \multicolumn{2}{c}{X Axis}\\
			\cline{2-5}
			& Before Optimization & After Optimization & Before Optimization & After Optimization\\
			\hline\hline
			$K_p$ & $30$ & $25.1221$ & $30$ & $23.1781$ \\
			
			$T_i$ & $0.002$ & $2.8459\times10^{-4}$ & $0.002$ & $3.8538\times10^{-4}$ \\
			
			$T_d$ & $0.00012$ & $1.3490\times10^{-4}$ & $0.00012$ & $2.3865\times10^{-4}$ \\
			
			\hline
		\end{tabular}
	\end{center}
	\label{table:para}
\end{table*}

\begin{table*}[t]
	\caption{Overview of the cost functions for Y and X axis}
	\begin{center}
		\begin{tabular}{ c|c|c|c|c }
			\multirow{2}{*}{Cost functions} &  \multicolumn{2}{c|}{Y Axis} & \multicolumn{2}{c}{X Axis}\\
			\cline{2-5}
			& Before Optimization & After Optimization & Before Optimization & After Optimization\\
			\hline\hline
			Total cost $J$ & $1.3892\times10^{8}$ & $2.4833\times10^{7}$ & $2.8850\times10^{6}$ & $6.1624\times10^{5}$ \\
			
			Tracking cost $J_e$ & $1.0890\times10^{8}$ & $5.6159\times10^{6}$ & $2.3326\times10^{6}$ & $3.9377\times10^{5}$\\
			
			Control variation cost $J_{\dot u}$ & $3.0017\times10^{7}$ & $1.9217\times10^{7}$ & $5.5240\times10^{5}$ & $2.2247\times10^{5}$\\
			\hline
		\end{tabular}
	\end{center}
	\label{table:cost}
\end{table*}

\begin{figure}
	\centering
	\includegraphics[width=0.45\textwidth]{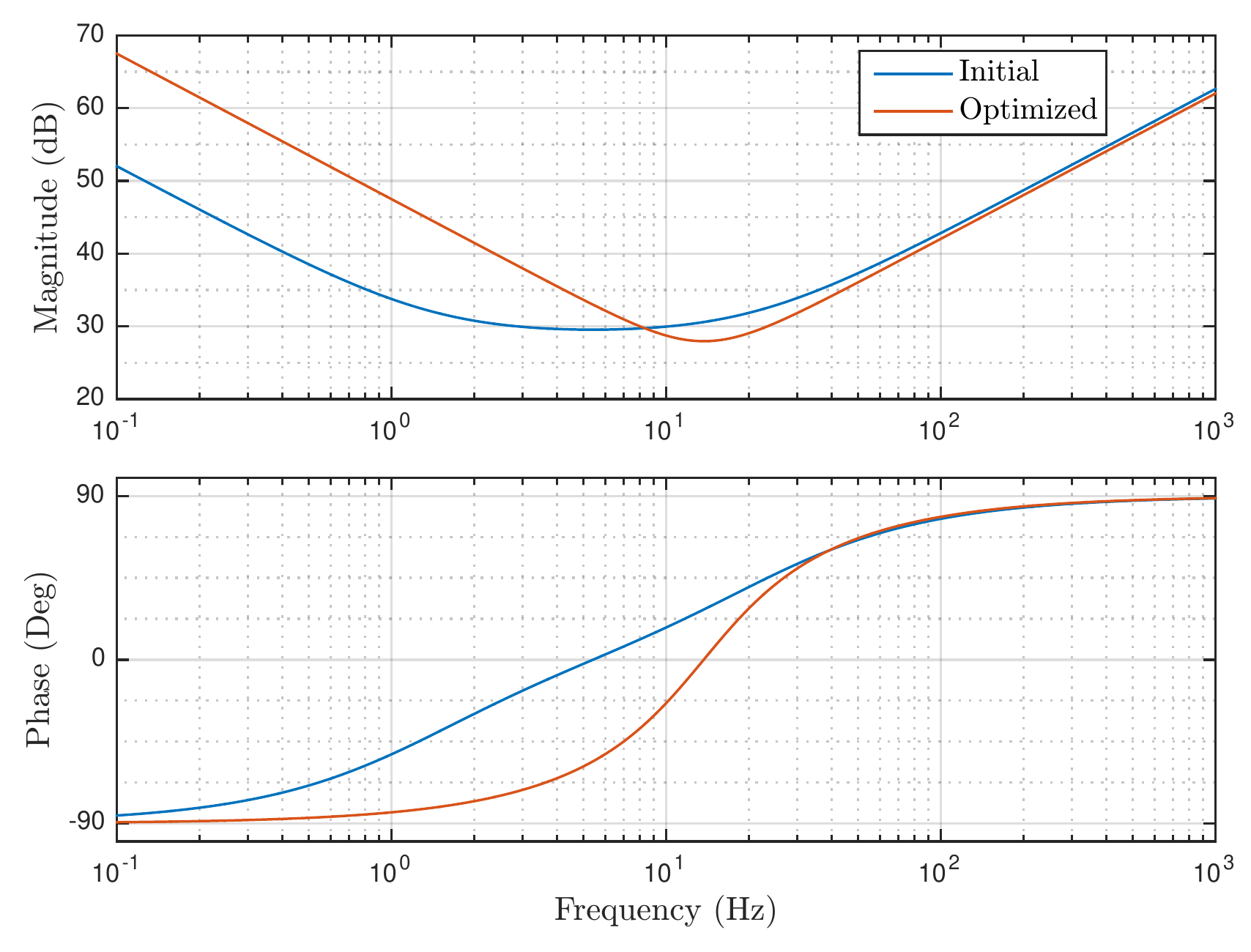}\\
	\caption{Y axis controller $C(s,\rho)$ comparison before and after the optimization in the frequency domain.}\label{fig:Controller}
\end{figure}
\begin{figure}
	\centering
	\includegraphics[width=0.43\textwidth]{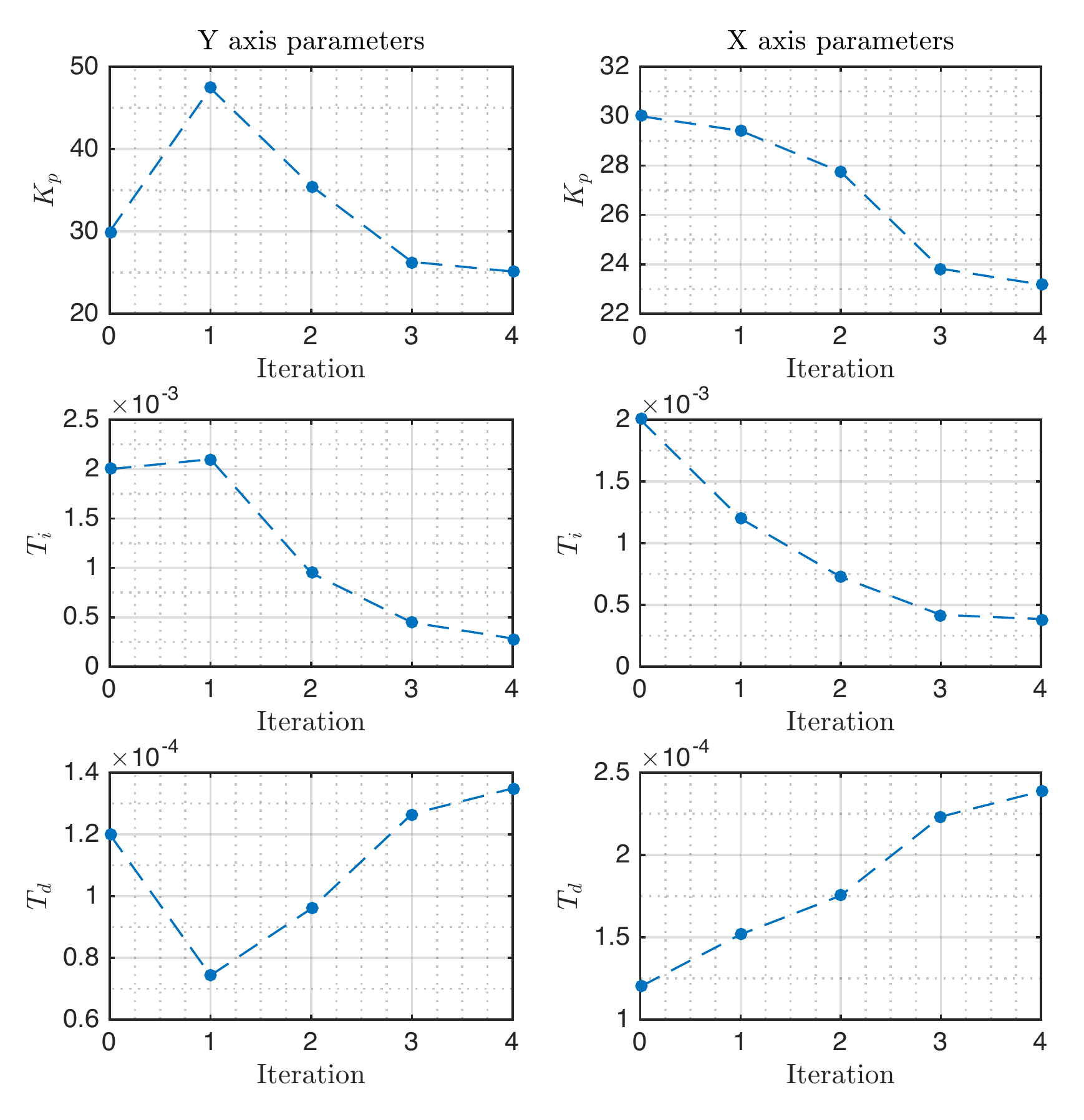}\\
	\caption{Y axis and X axis controller parameter convergence diagram.}\label{fig:XYPara}
\end{figure}
\begin{figure}
	\centering
	\includegraphics[width=0.43\textwidth]{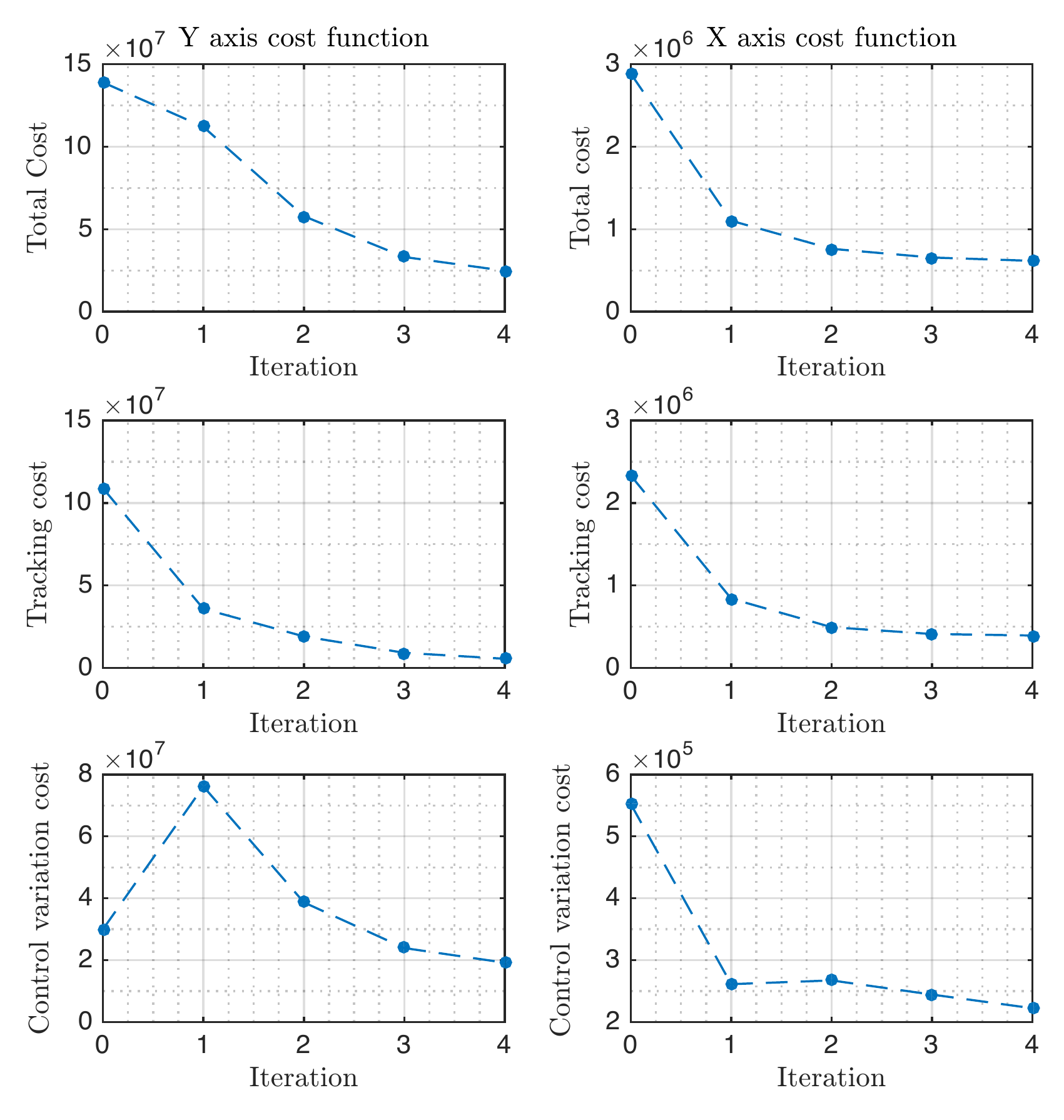}\\
	\caption{Y axis and X axis cost function convergence diagram. Top: Overall Cost. Middle: Cost related to tracking error. Bottom: Cost related to control signal variation}\label{fig:XYCost}
\end{figure}
\begin{figure}
	\centering
	\includegraphics[width=0.45\textwidth]{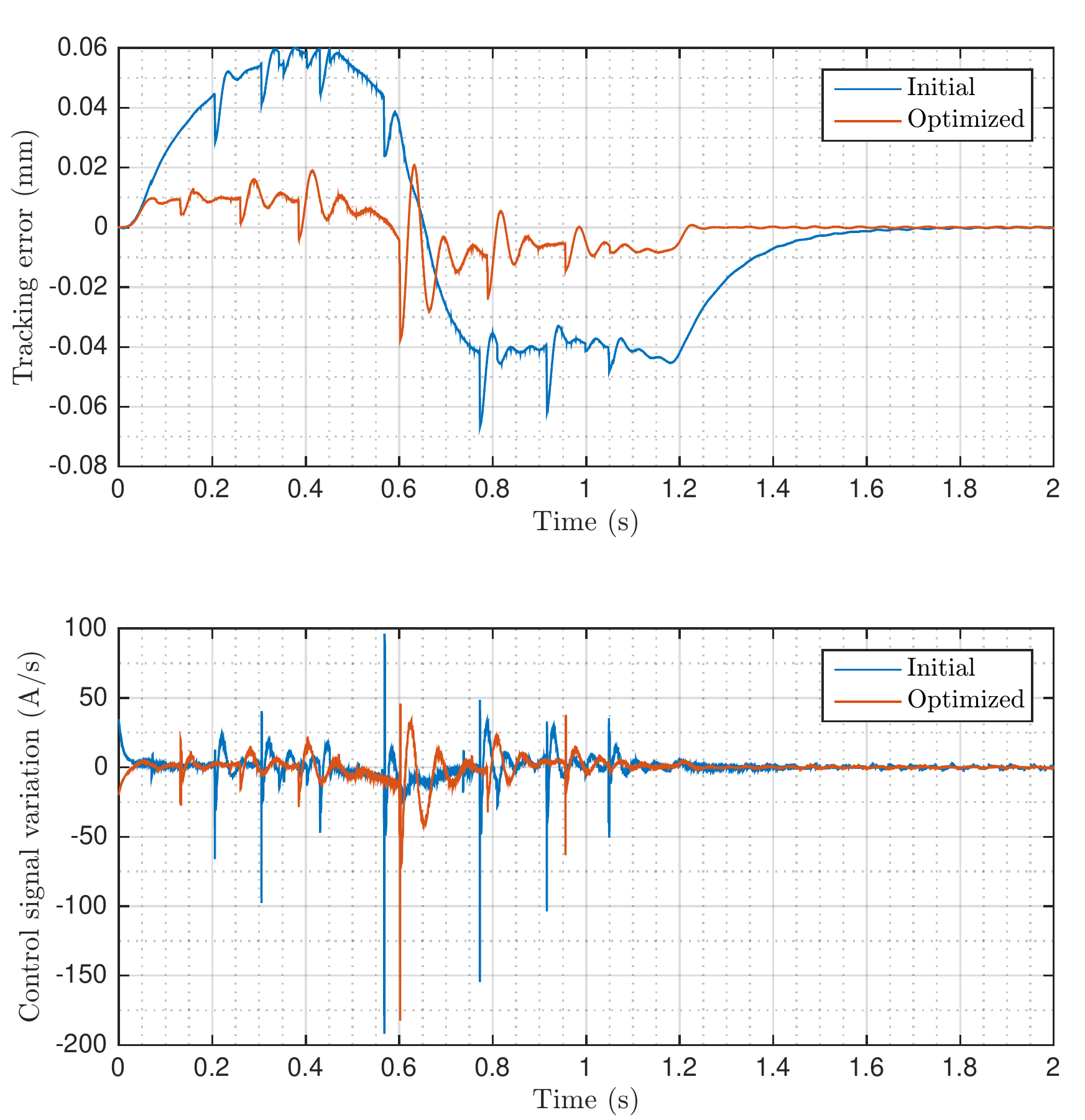}\\
	\caption{Y axis tracking error and control signal variation comparison before and after the data-driven multi-objective optimization.}\label{fig:YTracking}
\end{figure}
\begin{figure}
	\centering
	\includegraphics[width=0.45\textwidth]{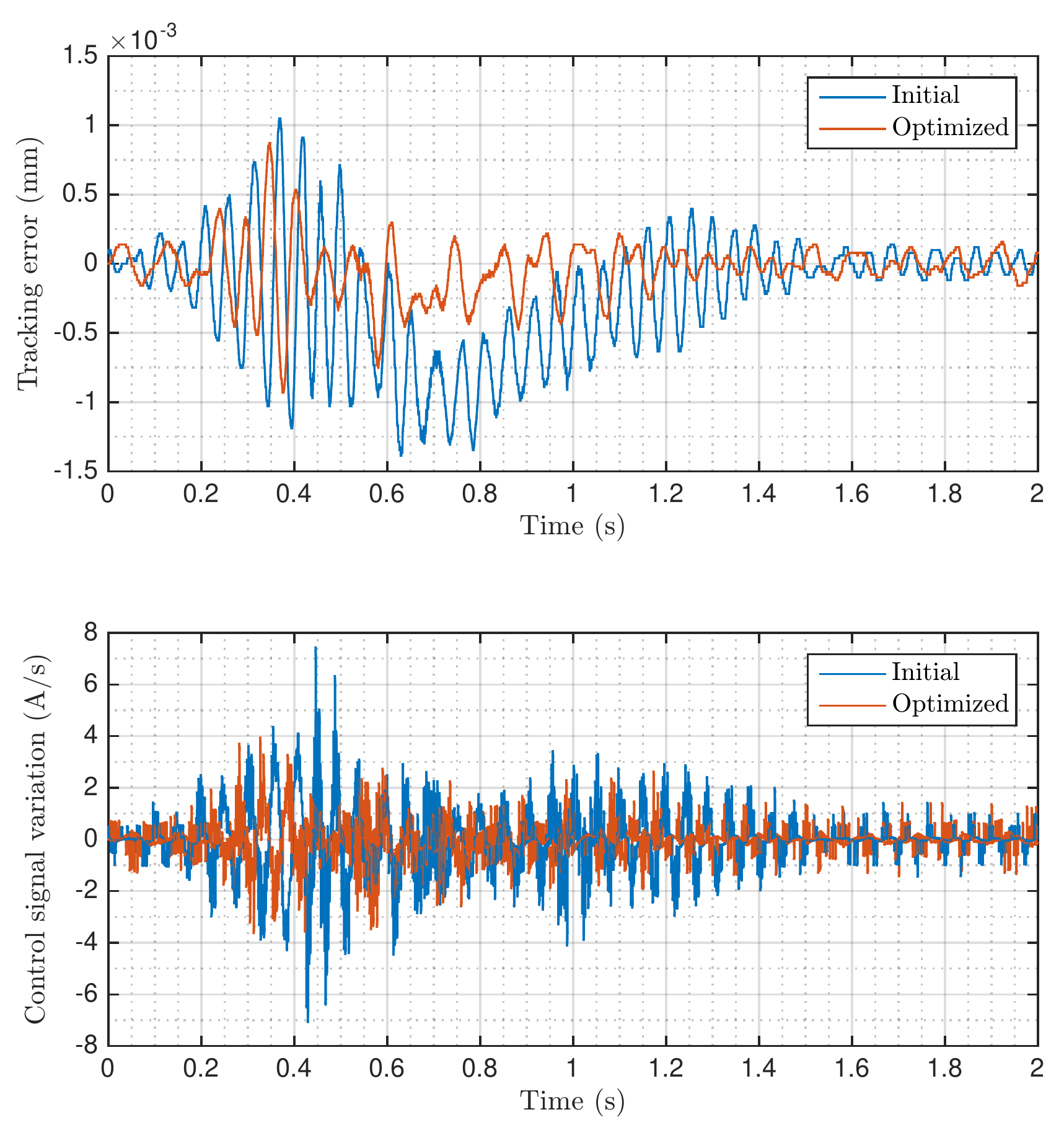}\\
	\caption{X axis tracking error and control signal variation comparison before and after the data-driven multi-objective optimization.}\label{fig:XTracking}
\end{figure}
\begin{figure}
	\centering
	\includegraphics[width=0.44\textwidth]{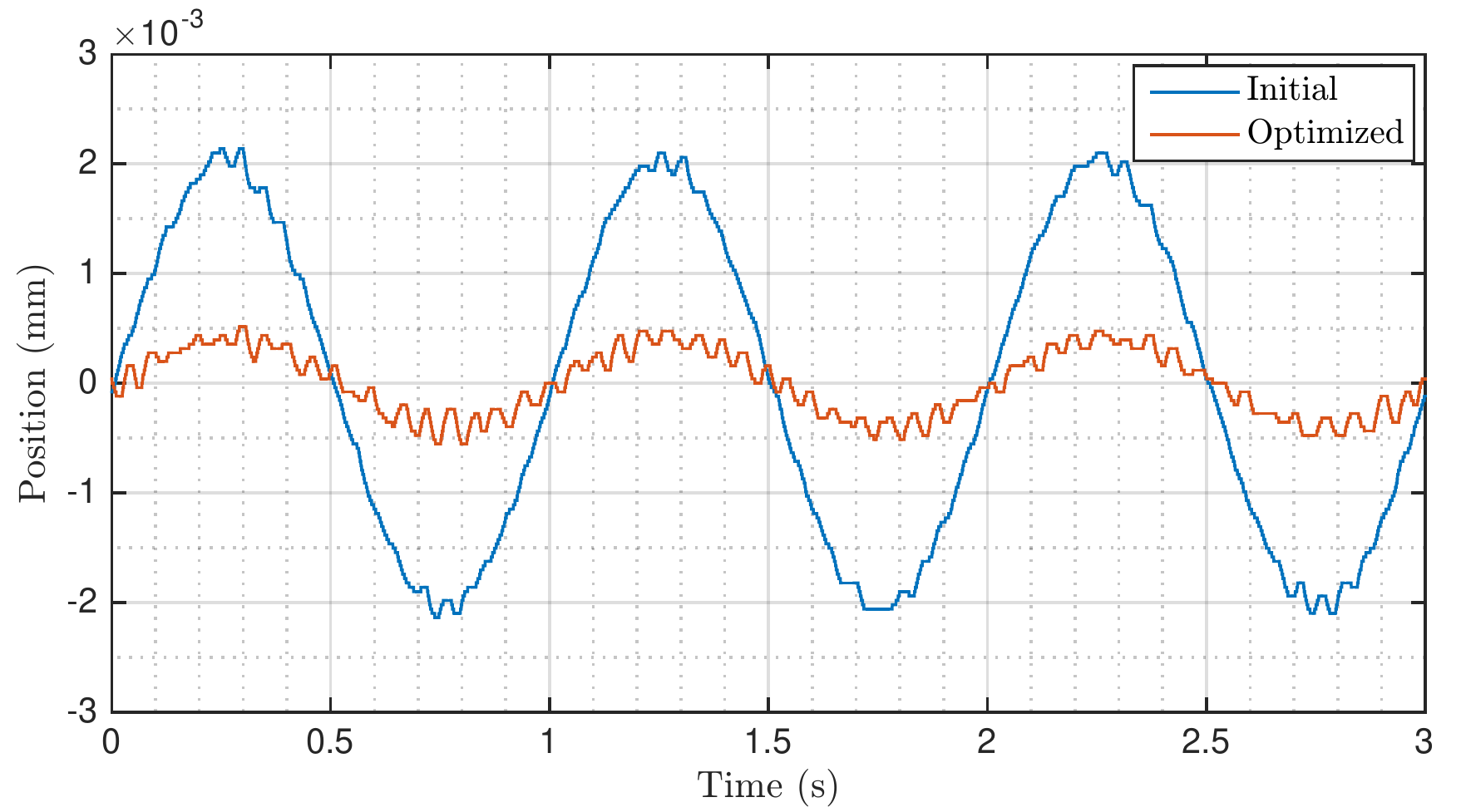}\\
	\caption{X axis disturbance rejection performance comparison under 1 Hz sinusoidal disturbance.}\label{fig:Disturbance}
\end{figure}

This section documents the experimental results of using the proposed data-driven optimization algorithm
for the maglev nanopositioning system as a case study.
A National Instruments (NI) PXI-8110 real-time controller is used with two FPGAs (NI PXI-7854R and 7831R)
to provide the necessary input/output (I/O) functions.
Two Trust TA320 and two TA115 linear current amplifiers are utilized
to power up the eight-phase coils.
The sampling frequency is 5 kHz, and the current limit for each phase of the coil arrays is set as 1.2 A.
The Renishaw fiber optic laser interferometers (Model: RLU10) are used for sensing of horizontal positions
with a count resolution of 40 nm,
and Lion Precision capacitive sensors (Model: CPL290 controller with C18 heads) are used for sensing of vertical positions
with a root mean square resolution of 150 nm.
The magnetically-levitated system including its actuation and sensor system are shown in Fig. \ref{fig:Maglevpic}, and its designed working range is 30\,mm\,$\times$\,30\,mm\,$\times$\,3\,mm according to the coil array length.

The motion profile used in the experiment is a fourth-order S-curve which is particularly suitable for precision motion control \cite{lambrechts2005trajectory}.
In order to meet the requirement of smooth motion, the profile is defined up to the fourth order with limited jerk and snap.
The position trajectory as well as its velocity, acceleration, jerk (time derivative of acceleration)
and snap (time derivative of jerk) are plotted in Fig. \ref{fig:Scurve}.
The magnetically-levitated system is controlled by a feedback controller in LabVIEW designed according to the typical proportional–integral–derivative (PID) structure,
as the PID control is essentially the most widely adopted control structure in the industry.
Nevertheless, it is pertinent at this juncture to also point out that
the data-driven multi-objective optimization algorithm proposed here
is also applicable to other types of feedback controllers,
as long as that it can be expressed in the rather common and standard form of \eqref{eq:Fixed_structure_C}.
Here specifically, the control input $u(t)$ is
\begin{equation}
u(t)=K_p(e(t)+\frac{1}{T_i} \int_{0}^{t} e(t^\prime) dt^\prime+T_d\frac{de(t)}{dt}),
\end{equation}
and the feedback controller can be written in the form of \eqref{eq:Fixed_structure_C} as
\begin{equation}
C(s,\rho)=\rho^T \bar C(s)=\begin{bmatrix}
K_p & K_p/T_i & K_pT_d \\
\end{bmatrix}
\begin{bmatrix}
1 \\
1/s \\
s
\end{bmatrix}.
\end{equation}
The goal of the data-driven optimization is to find out the controller parameters
that provide a smooth and accurate tracking of the motion profile
in Fig. \ref{fig:Scurve}, i.e., minimizing the cost function $J(\rho)$ in \eqref{eq:J}. Note that during the optimization process, no \textit{a priori} dynamic model information is needed nor will the algorithm attempt to build a model through system identification. To start with, the initial set of controller parameters ${^\mathbf{0}\rho}$ is designed based on the loop shaping method in \cite{zhu2017analysis} with a second order model (neglecting the nonlinearities and higher order dynamics) and further fine-tuned to provide a decent but non-optimized control performance, as in Table \ref{table:para}. It is worth noting, however, that loop shaping is a model-based method one can choose to use for the controller initialization but it is by no means necessary when there are no models available. In such cases, one shall simply tune the controller manually to achieve a decent performance and then rely on the proposed data-driven algorithm for performance optimization. The weightings are set as $w_1=10^7$ and $w_2=1$ so that the cost function values for the tracking error and control signal variation are on the same scale (the tracking error has a much smaller numerical value compared with the control signal variation). Nevertheless, we can still adjust the weightings according to the requirement of the motion system, i.e., further improvement on the accuracy or motion smoothness.

Despite the fact that the magnetically-levitated system is capable of conducting 6-DOF motion,
we consider here only the X-Y plane motion because it is most commonly used in semiconductor manufacturing \cite{yuan2016time,ma2017integrated}. Yet nevertheless, even in this application scenario, it is still the situation where the fully floating behavior and multi-axis coupling make extremely accurate identification of the motion dynamics largely impossible, so that traditional model-based approaches would encounter great difficulties
in being properly successfully deployed here. In high precision semiconductor manufacturing applications, it is often required to conduct a series of repetitive motions \cite{tan2019disturbance,li2018feedforward} on one of the axes.
Meanwhile, in order to guarantee the accuracy of highly complex semiconductor circuit patterns,
the tracking error from both Y and X axes needs to be minimized. Also, smooth motion should be ensured by minimizing the control input variation for both axes and using a higher-order S-curve trajectory.
By using the proposed data-driven multi-objective optimization in \textbf{Algorithm 1},
the control parameters in both Y and X axes are iteratively optimized as shown in Table \ref{table:para}
and the control performance in terms of the cost function
can be significantly improved as shown in Table \ref{table:cost}. In addition, a comparison of the optimized controller with the initial controller in the frequency domain is plotted in Fig. \ref{fig:Controller}.
One major advantage of this data-driven approach is its fast convergence rate
because it takes into account not only the gradient $\nabla J(^\mathbf{i}\rho)$ of the cost function
but the Hessian $\nabla^2 J(^\mathbf{i}\rho)$.
From Fig. \ref{fig:XYPara} and Fig. \ref{fig:XYCost},
we can observe that both the controller parameters and cost function value converge within only four iterations.
Note that the tracking cost $J_e$ or control variation cost $J_{\dot u}$ alone
may increase in some iteration, e.g., $J_{\dot u}$ of the Y axis in the $1^{\mathrm{st}}$ iteration,
but the total cost $J$ always decreases iteration by iteration.
The time-domain performance improvement for Y axis before
and after the data-driven optimization is plotted in Fig. \ref{fig:YTracking},
showing a significant reduction in both the tracking error and control signal variation; the root-mean-square (RMS) tracking errors are respectively $0.033$ mm and $0.0075$ mm. Here, the tracking error peaks, e.g., at $0.2$ s of the initial result, could be due to the laser interferometer signal loss or computational delays. Meanwhile, the tracking error and control signal variation of the X axis are also reduced as shown in Fig. \ref{fig:XTracking}. The RMS tracking errors for X axis are respectively $4.8297\times 10^{-4}$ mm and $2.0365\times 10^{-4}$ mm. Here, the tracking error is much smaller compared with the Y axis,
because the X axis is kept stationnary while the Y axis is moving. After $1.5$ s, the vibration still exists and this is due to the fact the stage is fully floating with little damping and has to deal with the disturbances from the force coupling. To further demonstrate the disturbance rejection performance, Fig.\mbox{ \ref{fig:Disturbance}} shows the X axis position measurement comparison under the effect of a $1$ Hz sinusoidal disturbance. From all these experimental results, we can see that the proposed approach is certainly effective and able to provide the appropriate optimized trajectory tracking in terms of both accuracy and smoothness, and the convergence rate is also suitably fast enough (only 4 iterations in our experiment) for practical applicability.

\section{Conclusion}

In this paper, we present a data-driven multi-objective optimization algorithm for repetitive motion tasks,
where no \textit{a priori} model information is available.
The proposed algorithm is applied to a multi-axis magnetically-levitated system
which is difficult to model accurately due to its fully floating behavior and multi-axis coupling.
By making use of the rich information contained in the actual motion data
under, say, the prevailing non-optimal conditions,
the algorithm can provide fast, efficient and effective controller optimization for the system
to operate towards optimality in a data-driven and iterative manner.
A well-designed cost function
is stated and specified, which
takes both smooth and accurate tracking into account
and the optimization process can be completed within a few iterations.
Our experimental results show that the motion performance of the maglev nanopositioning system
is enhanced significantly and could certainly meet the stringent requirement of present-day high-performance precision motion applications.

For future work, we believe applications of the proposed algorithm certainly
is not be limited to such a multi-axis magnetically-levitated system only,
and its potential
can be further exploited and deployed in other robotic systems (essentially especially those that are challenging to accurately model, e.g. quadrotors, legged robots, and soft robots, etc.).

\bibliographystyle{IEEEtran}
\bibliography{IEEEabrv,mybibTmech}
\begin{IEEEbiography}[{\includegraphics[width=1in,height=1.25in,clip,keepaspectratio]{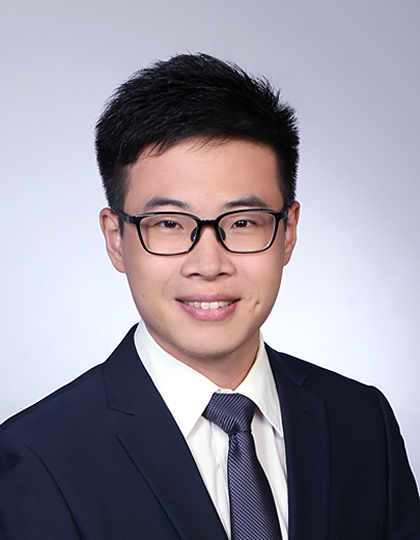}}]
	{Xiaocong Li} (S'14--M'17) received the B.Eng. degree from the Department of Electrical and Computer Engineering, National University of Singapore, in 2013, and the Ph.D. degree in electrical engineering from the NUS Graduate School for Integrative Sciences and Engineering, National University of Singapore, in 2017.

    He is a Research Scientist with the Mechatronics Group, Singapore Institute of Manufacturing Technology (SIMTech), Agency for Science, Technology and Research (A*STAR), where he is currently leading a Collaborative Research Project (CRP) with the National University of Singapore on data-driven controls. Since 2018, he has served as a member of the Thesis Advisory Committee for NUS Graduate School for Integrative Science and Engineering. His current research interest is focused on making sense of data for controls through learning as well as its applications to robotics and precision motion systems.
\end{IEEEbiography}

\begin{IEEEbiography}[{\includegraphics[width=1in,height=1.25in,clip,keepaspectratio]{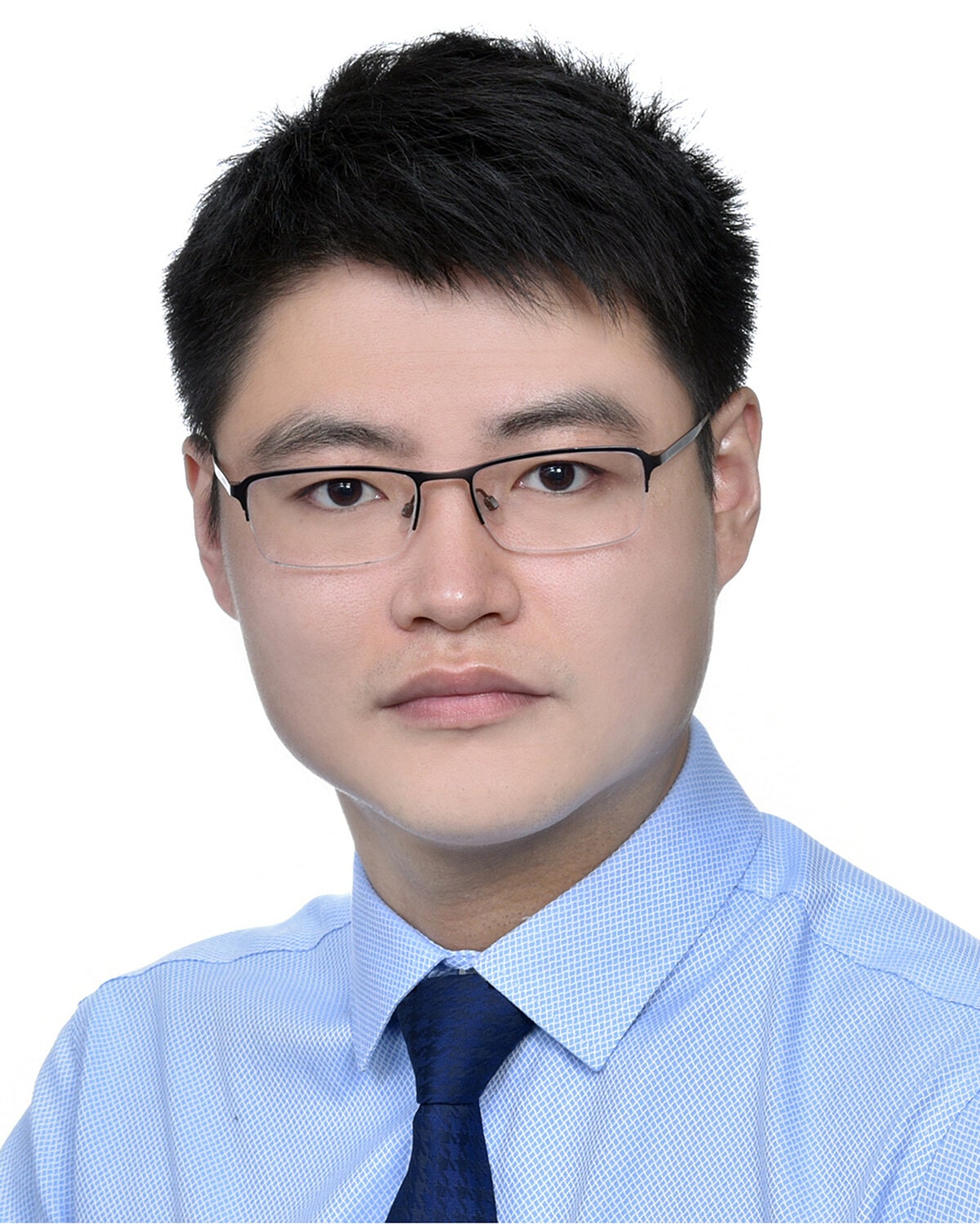}}]
	{Haiyue Zhu} (S'13--M'17) received the B.Eng. degree in automation from the School of Electrical Engineering and Automation and the B. Mgt. degree in business administration from the College of Management and Economics, Tianjin University, Tianjin, China, in 2010, and the M.Sc. and Ph.D. degrees from the National University of Singapore (NUS), Singapore, in 2013 and 2017, respectively, both in electrical engineering.
	
	He is currently a Scientist in Singapore Institute of Manufacturing Technology (SIMTech), Agency for Science, Technology and Research (A*STAR). His current research interests include intelligent mechatronic and robotic systems, etc.
\end{IEEEbiography}

\begin{IEEEbiography}[{\includegraphics[width=1in,height=1.25in,clip,keepaspectratio]{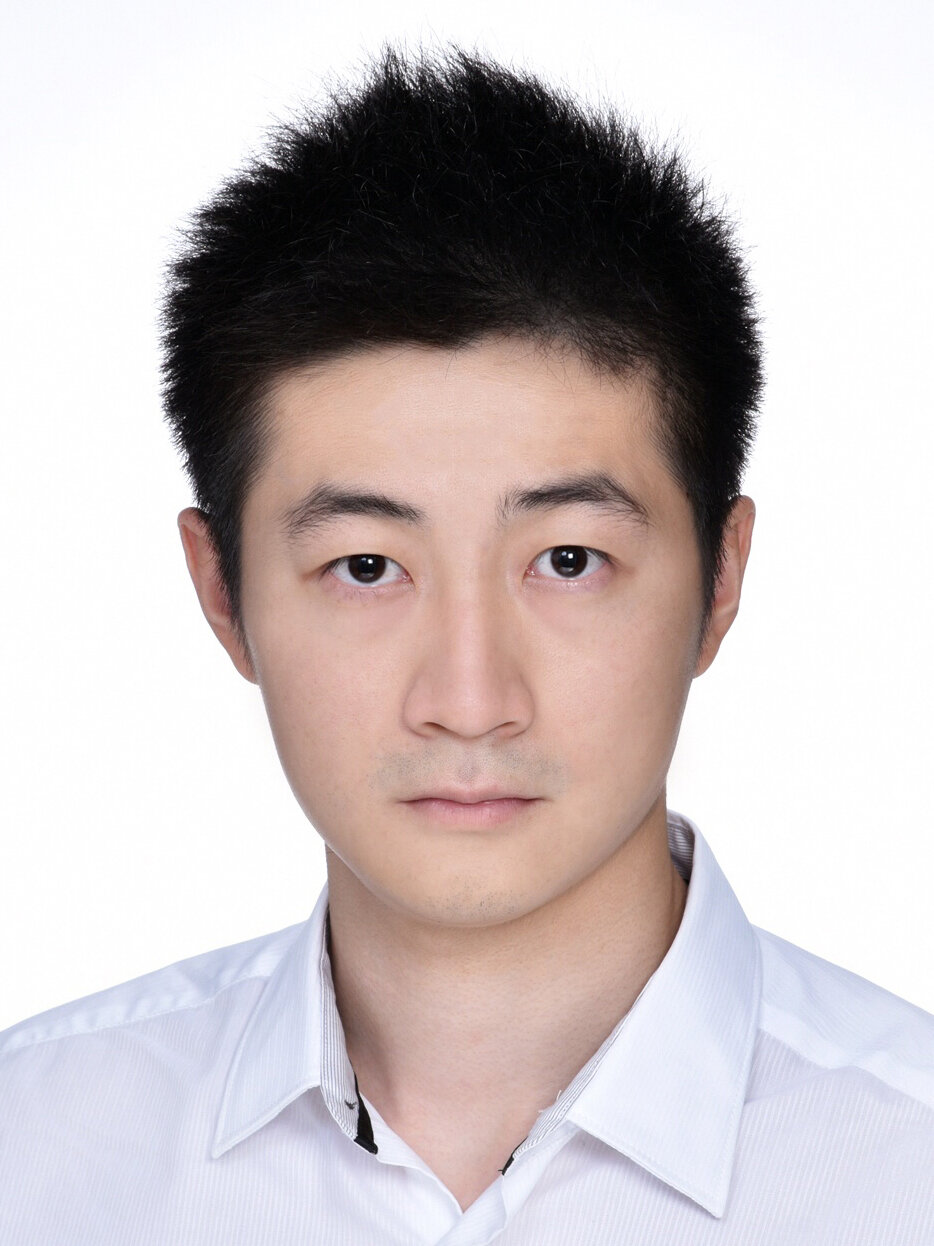}}]{Jun Ma}
	(S'15-M'18) received the B.Eng. degree with First Class Honours in electrical and electronic engineering from the Nanyang Technological University, Singapore, in 2014, and the Ph.D. degree in electrical and computer engineering from the National University of Singapore, Singapore, in 2018.
	
	From 2018 to 2019, he was a Research Fellow with the Department of Electrical and Computer Engineering, National University of Singapore, Singapore. In 2019, he was a Research Associate with the Department of Electronic and Electrical Engineering, University College London, London, U.K. He is currently a Visiting Scholar with the Department of Mechanical Engineering, University of California, Berkeley, Berkeley, CA, USA. His research interests include control and optimization, precision mechatronics, robotics, and medical technology.
	
	He was a recipient of the Singapore Commonwealth Fellowship in Innovation.
\end{IEEEbiography}

\begin{IEEEbiography}[{\includegraphics[width=1in,height=1.25in,clip,keepaspectratio]{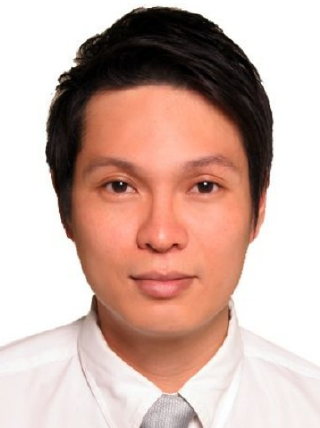}}]
	{Tat Joo Teo} (M'08) received the B. Eng. degree in mechatronics engineering from Queensland University of Technology, Australia, in 2003 and the Ph.D. degree from Nanyang Technological University, Singapore, in 2009.
	
	He was a research scientist with the Singapore Institute of Manufacturing Technology, Singapore, from 2009 to 2018. He was a visiting scientist with Massachusetts Institute of Technology, USA, in 2016, and also served on the engineering faculty at the National University of Singapore and Newcastle University in Singapore. His research interest is to explore the fundamentals of Newtonian mechanics, solid mechanics, kinematics, and electromagnetism to develop high precision mechatronics or robotic systems for micro-/nano-scale manipulation and bio-medical applications.
	
	Dr. Teo has published over 50 peer-reviewed articles and has 4 patents granted. In 2013, he received the IECON Best Paper Award in the theory and servo design category. In 2014, he became the first Singaporean to win the R\&D 100 Award, which is the most prestigious international award for technologically-significant products.
	
\end{IEEEbiography}

\begin{IEEEbiography}[{\includegraphics[width=1in,height=1.25in,clip,keepaspectratio]{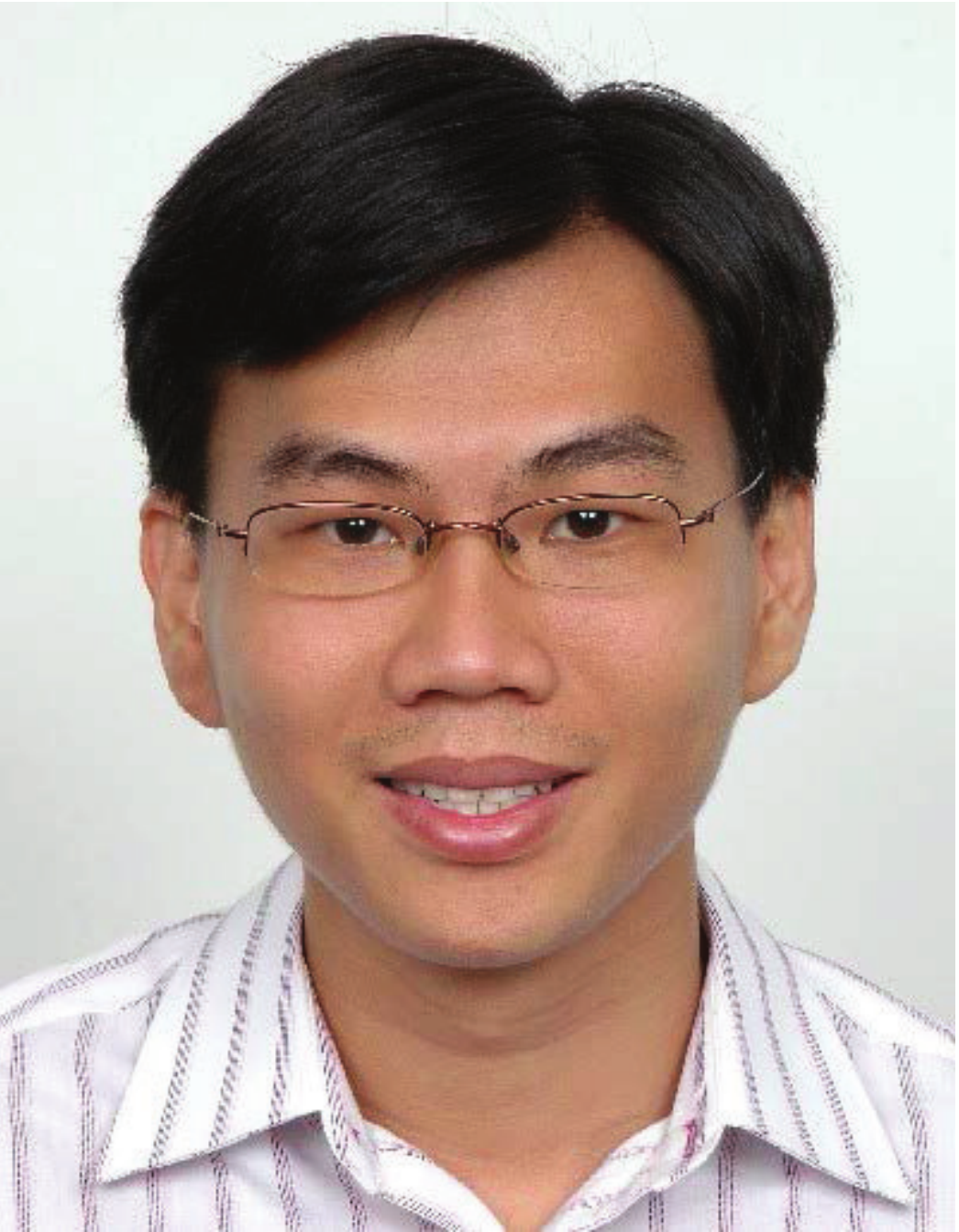}}]{Chek Sing Teo}(S'04-M'08)
	completed his Ph.D. degree at the National University of Singapore in 2008, under the Agency for Science Technology and Research (A-STAR) Scholarship Scheme, working on "Accuracy Enhancement for High Precision Gantry Stage". His research interests are in the application of advanced control techniques to precision mechatronic system and instrumentation; to enhance performance in motion control and measurement. His current work includes using mechatronics stiffness to reduce jerk reaction in high speed motion stage and sensor placement for adaptronics. He is currently working in the Singapore Institute of Manufacturing Technology (SIMTech) leading the Precision Mechatronics Team within the Mechatronics Group, as well as the co-Director of the SIMTech-NUS Precision Motion System joint lab.
\end{IEEEbiography}

\begin{IEEEbiography}[{\includegraphics[width=1in,height=1.25in,clip,keepaspectratio]{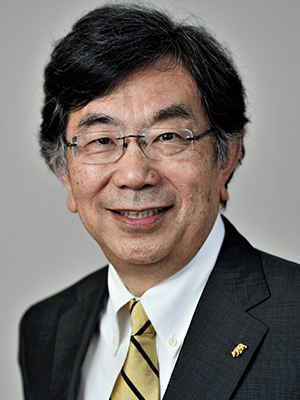}}]{Masayoshi Tomizuka}
	(M'86-SM'95-F'97-LF'17) received the B.S. the and M.S. degrees in mechanical engineering from the Keio University, Tokyo, Japan, and the Ph.D. degree in mechanical engineering from the Massachusetts Institute of Technology in February 1974.
	
	In 1974, he joined the faculty of the Department of Mechanical Engineering at the University of California at Berkeley, where he currently holds the Cheryl and John Neerhout, Jr., Distinguished Professorship Chair. His current research interests are optimal and adaptive control, digital control, motion control, and their applications to robotics and vehicles.
	
	He served as Program Director of the Dynamic Systems and Control Program of the Civil and Mechanical Systems Division of NSF (2002-2004). He served as Technical Editor of the ASME Journal of Dynamic Systems, Measurement and Control, J-DSMC (1988-93) and Editor-in-Chief of the IEEE/ASME Transactions on Mechatronics (1997-99). He is a Life Fellow of the ASME and IEEE, and a Fellow of International Federation of Automatic Control (IFAC) and the Society of Manufacturing Engineers. He is the recipient of the ASME/DSCD Outstanding Investigator Award (1996), the Charles Russ Richards Memorial Award (ASME, 1997), the Rufus Oldenburger Medal (ASME, 2002), the John R. Ragazzini Award (2006), and the Richard Bellman Control Heritage Award (2018).
\end{IEEEbiography}

\begin{IEEEbiography}[{\includegraphics[width=1in,height=1.25in,clip,keepaspectratio]{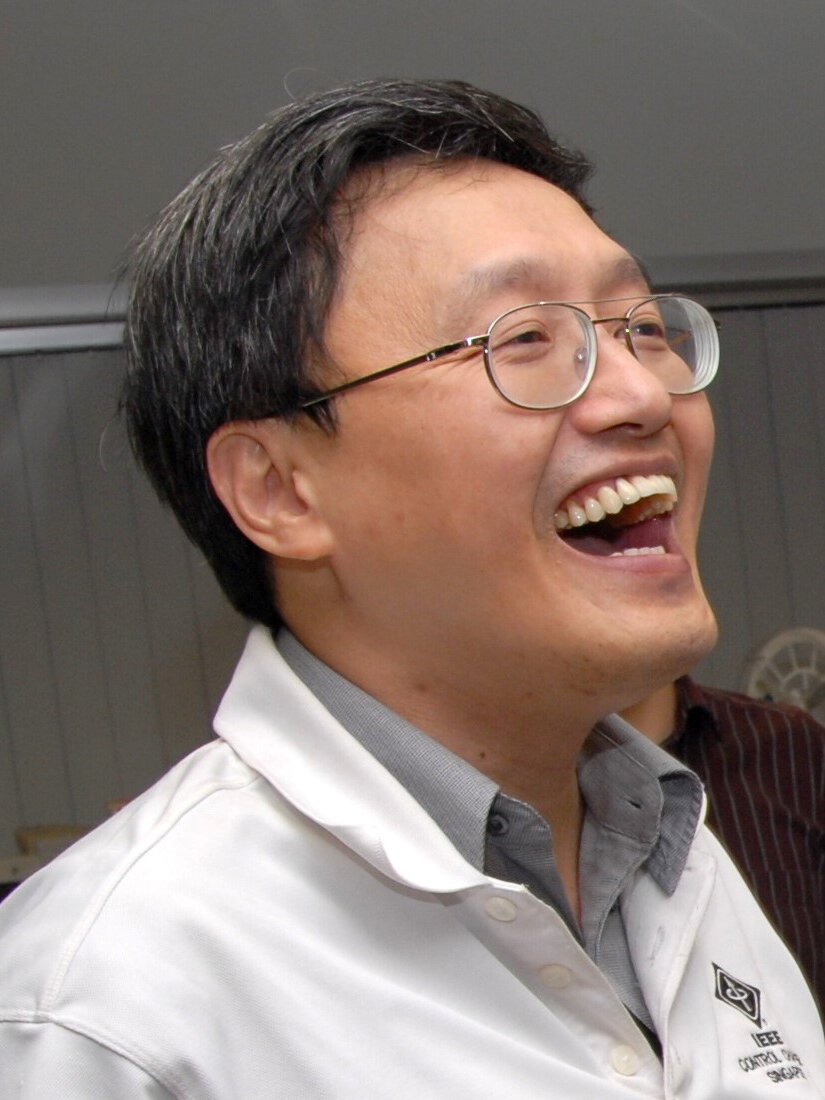}}]{Tong Heng Lee}
	received the B.A. degree with First Class Honours in the Engineering Tripos from the Cambridge University, Cambridge, U.K., in 1980, the M.Engrg. degree from the National University of Singapore (NUS), Singapore, in 1985, and the Ph.D. degree from the Yale University, New Haven, CT, USA, in 1987.
	
	He is a Professor in the Department of Electrical and Computer Engineering at the National University of Singapore, and also a Professor in the NUS Graduate School, NUS NGS. He was a Past Vice-President (Research) of NUS. Dr. Lee's research interests are in the areas of adaptive systems, knowledge-based control, intelligent mechatronics, and computational intelligence.
	
	He currently holds Associate Editor appointments in the IEEE Transactions in Systems, Man and Cybernetics, Control Engineering Practice (an IFAC journal), and the International Journal of Systems Science (Taylor and Francis, London). In addition, he is the Deputy Editor-in-Chief of IFAC Mechatronics journal.
\end{IEEEbiography}

\end{document}